\documentclass[journal]{IEEEtran}

\ifCLASSINFOpdf
\else
\fi

\usepackage[cmex10]{amsmath}

\usepackage{amssymb,amsthm}
\usepackage{color}
\usepackage{graphics}
\usepackage{tikz}
\usetikzlibrary{decorations.pathreplacing}
\usetikzlibrary{shapes.misc, fit}

\usepackage{multirow}
\usepackage{hhline}

\usepackage[linesnumbered,ruled,vlined,titlenumbered]{algorithm2e}
\usepackage{algorithmicx}
\usepackage[noend]{algpseudocode}

\allowdisplaybreaks

\usepackage[noadjust]{cite}

\theoremstyle{definition}
\newtheorem{definition}{Definition}

\newtheorem{example}{Example}

\theoremstyle{plain}
\newtheorem{theorem}{Theorem}
\newtheorem{proposition}{Proposition}
\newtheorem{lemma}{Lemma}
\newtheorem{corollary}{Corollary}

\DeclareMathOperator{\rk}{rk}
\DeclareMathOperator{\frk}{frk}
\DeclareMathOperator{\diag}{Diag}

\DeclareMathOperator{\dd}{d}
\DeclareMathOperator{\wt}{wt}

\begin{document}
%
\title{Maximum Sum-Rank Distance Codes over\\Finite Chain Rings}
%
%
%
\author{Umberto Mart{\'i}nez-Pe\~{n}as, 
and Sven Puchinger,
\thanks{This work was done in part while U. Mart{\'i}nez-Pe\~{n}as was with the Institutes of Computer Science and Mathematics, University of Neuch{\^a}tel, Switzerland. He is now with the University of Valladolid, Spain. During this period he was supported by a Mar{\'i}a Zambrano contract by the University of Valladolid, Spain, Contract no. E-47-2022-0001486 and by Grant TED2021-130358B-I00 funded by MCIN/AEI/10.13039/501100011033 and by the ``European Union NextGenerationEU/PRTR'' (e-mail: umberto.martinez@uva.es). }
\thanks{This work was done while S. Puchinger was with the Department of Applied Mathematics and Computer Science, Technical University of Denmark (DTU), Lyngby, Denmark and the Department of Electrical and Computer Engineering, Technical University of Munich, Munich, Germany. Within this period he was supported by the European Union's Horizon 2020 research and innovation program under the Marie Sklodowska-Curie grant agreement no. 713683 and by the European Research Council (ERC) under the European Union's Horizon 2020 research and innovation programme, grant agreement no. 801434 (e-mail: mail@svenpuchinger.de).}
\thanks{Parts of this paper have been presented at The Twelfth International Workshop on Coding and Cryptography (WCC 2022), Rostock, Germany, March, 2022 \cite{wcc_version}.}
}


%
%

\markboth{}%
{}
%



\maketitle



\begin{abstract}
In this work, maximum sum-rank distance (MSRD) codes and linearized Reed-Solomon codes are extended to finite chain rings. It is proven that linearized Reed-Solomon codes are MSRD over finite chain rings, extending the known result for finite fields. For the proof, several results on the roots of skew polynomials are extended to finite chain rings. These include the existence and uniqueness of minimum-degree annihilator skew polynomials and Lagrange interpolator skew polynomials. A general cubic-complexity sum-rank Welch-Berlekamp decoder and a quadratic-complexity sum-rank syndrome decoder (under some assumptions) are then provided over finite chain rings. The latter also constitutes the first known syndrome decoder for linearized Reed--Solomon codes over finite fields. Finally, applications in Space-Time Coding with multiple fading blocks and physical-layer multishot Network Coding are discussed.
\end{abstract}


%
\IEEEpeerreviewmaketitle

\section{Introduction} \label{sec intro}
%
%
%
%

\IEEEPARstart{T}{he} sum-rank metric, introduced in \cite{multishot}, is a natural generalization of both the Hamming metric and the rank metric. Codes considered with respect to the sum-rank metric over finite fields have applications in multishot Network Coding \cite{multishot, secure-multishot}, Space-Time Coding with multiple fading blocks \cite{space-time-kumar, mohannad} and local repair in Distributed Storage \cite{universal-lrc}. However, codes over rings may be more suitable for physical-layer Network Coding \cite{feng1}, where alphabets are subsets or lattices of the complex field instead of finite fields. Similarly, finite rings derived from the complex field allows for more flexible choices of constellations to construct Space-Time codes \cite{kiran, kamche}.

Maximum sum-rank distance (MSRD) codes are those codes whose minimum sum-rank distance attains the Singleton bound. Among known MSRD codes over finite fields, linearized Reed-Solomon codes \cite{linearizedRS} are those with smallest finite-field sizes (thus more computationally efficient) for the main parameter regimes, see \cite[Table 1]{generalMSRD} and \cite[Sec. 2.4]{generalMSRD}. Furthermore, they cover a wide range of parameter values and are the only known MSRD codes compatible with square matrices. Linearized Reed-Solomon codes include both generalized Reed-Solomon codes \cite{reed-solomon} and Gabidulin codes \cite{gabidulin}, whenever the sum-rank metric includes the Hamming metric and the rank metric, respectively. Reed-Solomon codes over rings were systematically studied for the first time in \cite{quintin}. Gabidulin codes over Galois rings were introduced in \cite{kiran}, and later extended to finite principal ideal rings in \cite{kamche}. Such families of Gabidulin codes over rings were proposed for Space-Time Coding in the case of a single fading block in \cite{kamche, kiran}, and they were proposed for physical-layer singleshot Network Coding in \cite{kamche}.

In this work, we introduce and study MSRD codes and linearized Reed-Solomon codes over finite chain rings, together with their applications in Space-Time Coding with multiple fading blocks and physical-layer multishot Network Coding. Finite chain rings are those finite rings whose family of ideals form a chain with respect to set inclusion. They are an important subfamily of finite principal ideal rings. In fact, finite principal ideal rings are all Cartesian products of finite chain rings \cite[Th. VI.2]{mcdonald}. Moreover, finite chain rings include Galois rings (although not all finite chain rings are Galois rings \cite[Th. XVII.5]{mcdonald}), which are of the form $ \mathbb{Z}_{p^r}[x] / (F) $, where $ p $ is a prime number, $ r $ is a positive integer, and $ F \in \mathbb{Z}_{p^r}[x] $ is a polynomial whose reduction modulo $ p $ is irreducible. Galois rings hence include finite fields (when $ r = 1 $) and finite rings of the form $ \mathbb{Z}_{p^r} $ (when $ F = x $). Finite chain rings also include quotients of subrings of the complex field of the form $ \mathbb{Z}_{2^r}[i] = \mathbb{Z}[i]/(2^r) $, where $ i $ is the imaginary unit, see \cite[Th. XVII.5]{mcdonald}. 

The contributions and organization of this manuscript are as follows. In Section \ref{sec preliminaries}, we collect some preliminaries on finite chain rings. In Section \ref{sec MSRD codes}, we define the sum-rank metric over finite chain rings, together with the corresponding Singleton bound and the definition of MSRD codes. Section \ref{sec skew polynomials} contains the theoretical building blocks for constructing linearized Reed-Solomon codes and their decoding. We extend Lam and Leroy's results \cite{lam, lam-leroy} relating roots and degrees of skew polynomials to the case of finite chain rings. As a result, we prove the existence and uniqueness of minimum-degree annihilator skew polynomials and Lagrange interpolator skew polynomials, and we describe when the corresponding extended Moore matrices are invertible. In Section \ref{sec LRS}, we introduce linearized Reed-Solomon codes over finite chain rings and use the previous results to prove that they are MSRD. In Section \ref{sec WB decoder}, we provide a cubic-complexity Welch-Berlekamp decoder with respect to the sum-rank metric for linearized Reed--Solomon codes over finite chain rings that works in general. Then, in Section \ref{sec syndrome decoder}, we provide a quadratic-complexity syndrome decoder with respect to the sum-rank metric for linearized Reed--Solomon codes over finite chain rings that work under some (not very strict) assumptions. This decoder also constitutes the first known syndrome decoder for linearized Reed--Solomon codes over finite fields, to the best of our knowledge. Finally, in Section \ref{sec applications}, we discuss applications in Space-Time Coding with multiple fading blocks and physical-layer multishot Network Coding.

We conclude by mentioning that there are other constructions of MSRD codes in the case of finite fields, in particular using different geometric points of view \cite{alberto-fundamental, generalMSRD, neri-oneweight, twisted}. However, we leave as an open problem generalizing them to finite chain rings.

\section*{Notation}

Let $ m $ and $ n $ be positive integers. We denote $ [n] = \{ 1,2, \ldots, n \} $. For a set $ \mathcal{A} $, we denote by $ \mathcal{A}^{m \times n} $ the set of $ m \times n $ matrices with entries in $ \mathcal{A} $, and we denote $ \mathcal{A}^n = \mathcal{A}^{1 \times n} $. All rings are considered with identity, and ring morphisms map identities to identities. Unless otherwise stated, rings are assumed to be commutative. For a ring $ R $, we denote by $ R^* $ the set of units of $ R $. For $ a \in R $, we denote by $ (a) $ the ideal generated by $ a $.

\section{Preliminaries on Finite Chain Rings} \label{sec preliminaries}

In this preliminary section, we introduce and revisit some important properties of finite chain rings. We refer the reader to \cite{mcdonald} for more details.

A \textit{local ring} is a ring with only one maximal ideal, and a \textit{chain ring} is a ring whose ideals form a chain with respect to set inclusion, thus being a local ring. Throughout this manuscript, we fix a finite chain ring $ R $, meaning a chain ring of finite size. We will denote by $ \mathfrak{m} $ the maximal ideal of $ R $. Since $ R $ is finite and $ R/\mathfrak{m} $ is a field, then it must be a finite field. We will fix the prime power $ q = |R/\mathfrak{m}| $, and we denote $ \mathbb{F}_q = R / \mathfrak{m} $, the \textit{finite field} with $ q $ elements.

Let $ H \in R[x] $ be a monic polynomial of degree $ m $ whose image in $ \mathbb{F}_q[x] $ is irreducible. Throughout this manuscript, we will fix $ S = R[x]/(H) $. The ring $ S $ is a free local Galois extension of $ R $ (hence a free $ R $-module) of rank $ m $ with maximal ideal $ \mathfrak{M} = \mathfrak{m} S $. Furthermore, the Galois group of $ R \subseteq S $ is cyclic of order $ m $, and generated by a ring automorphism $ \sigma : S \longrightarrow S $ such that $ R = \{ a \in S \mid \sigma(a) = a \} $ and $ \sigma(c) = c^q $, for some primitive element $ c \in S $. Moreover, it holds that $ S/\mathfrak{M} = \mathbb{F}_{q^m} $, and we have a commutative diagram
\begin{equation}
\begin{array}{rcl}
S & \stackrel{\sigma}{\longrightarrow} & S \\
\rho \downarrow & & \downarrow \rho \\
\mathbb{F}_{q^m} & \stackrel{\overline{\sigma}}{\longrightarrow} & \mathbb{F}_{q^m},
\end{array}
\label{eq commutative diagram}
\end{equation}
where $ \rho : S \longrightarrow S / \mathfrak{M} = \mathbb{F}_{q^m} $ is the natural projection map, and $ \overline{\sigma}(a) = a^q $, for all $ a \in \mathbb{F}_{q^m} $. In other words, $ \rho (\sigma(a)) = \overline{\sigma}(\rho(a)) $, for all $ a \in S $. We will usually denote $ \overline{a} = \rho(a) $, and therefore, we have that $ \overline{\sigma(a)} = \overline{\sigma}(\overline{a}) $, for $ a \in S $.

\begin{example} \label{ex first}
Let $ R = \mathbb{Z}_9 $, that is, the ring of integers modulo $ 9 $. It is clearly a finite chain ring with maximal ideal $ \mathfrak{m} = (3) $. Its residue field is the finite field $ R/\mathfrak{m} = \mathbb{F}_3 = \mathbb{Z}_3 $ and $ q = 3 $. We may choose $ H = x^2 + 1 $ (i.e., $ m = 2 $) and construct the finite residue ring $ S = R[x] / (H) = \mathbb{Z}_9[x] / (x^2+1) $. Denote by $ \alpha \in S $ the image of $ x $ in $ S $, which satisfies $ \alpha^2 + 1 = 0 $. The set $ S $ is then
$$ S = \{ a \alpha + b \mid a,b \in \mathbb{Z}_9 \}. $$
We may define the morphism $ \sigma : S \longrightarrow S $ given by $ \sigma(\alpha) = \alpha^3 $ and being the identity on $ \mathbb{Z}_9 $. It is well defined since $ (\alpha^3)^2+1 = 0 $ and it is an automorphism that generates the Galois group of $ S $ over $ R $ since $ \sigma^2 = {\rm Id} $ is the identity map. Notice that $ m = 2 $ in this case and we have the residue field $ S / \mathfrak{M} = \mathbb{F}_9 $.
\end{example}

An important feature of local rings is that the group of units is formed by the elements outside of the maximal ideal. That is, $ R^* = R \setminus \mathfrak{m} $ and $ S^* = S \setminus \mathfrak{M} $. As stated above, $ S $ is a free $ R $-module of rank $ m $, and any basis of $ S $ over $ R $ has $ m $ elements. Finally, the following technical lemma will be useful for our purposes.

\begin{lemma} \label{lemma linear independence}
Let $ \beta_1, \beta_2, \ldots, \beta_r \in S $ be $ R $-linearly independent (thus $ r \leq m $). 
\begin{enumerate}
\item
There exist $ \beta_{r+1}, \ldots, \beta_m \in S $ such that $ \beta_1, \beta_2, \ldots, \beta_m $ form a basis of $ S $ over $ R $.
\item
The projections $ \overline{\beta}_1, \overline{\beta}_2, \ldots, \overline{\beta}_r \in \mathbb{F}_{q^m} $ are $ \mathbb{F}_q $-linearly independent.
\item
$ \beta_1, \beta_2, \ldots, \beta_r \in S^* $.
\end{enumerate}
\end{lemma}
\begin{proof}
Item 1 is a particular case of \cite[p. 92, ex. V.14]{mcdonald}. Now, since $ \beta_1, \beta_2, \ldots, \beta_m $ are generators of $ S $ over $ R $, then $ \overline{\beta}_1, \overline{\beta}_2, \ldots, \overline{\beta}_m $ are generators of $ \mathbb{F}_{q^m} $ over $ \mathbb{F}_q $, thus they are a basis since there are $ m $ of them. Thus $ \overline{\beta}_1, \overline{\beta}_2, \ldots, \overline{\beta}_r $ are $ \mathbb{F}_q $-linearly independent. Finally, Item 3 is \cite[Lemma 2.4]{kamche} but is also trivial from Item 2 since $ S^* = S \setminus \mathfrak{M} $.
\end{proof}

\section{MSRD Codes on Finite Chain Rings} \label{sec MSRD codes}

The sum-rank metric over fields was first defined in \cite{multishot} under the name \textit{extended distance}, although it was previously used in the Space-Time Coding literature \cite[Sec. III]{space-time-kumar}. Later, the rank metric was extended to finite principal ideal rings in \cite{kamche}. In this section, we will introduce the sum-rank metric for finite chain rings.

Since $ R $ is a finite chain ring, then it is a principal ideal ring. Therefore, given a matrix $ \mathbf{A} \in R^{m \times n} $, there exist two invertible matrices $ \mathbf{P} \in R^{m \times m} $ and $ \mathbf{Q} \in R^{n \times n} $, and a (possibly rectangular) diagonal matrix $ \mathbf{D} = \diag(d_1, d_2, \ldots, d_r) \in R^{m \times n} $, with $ r = \min \{m,n\} $, such that $ \mathbf{A} = \mathbf{P} \mathbf{D} \mathbf{Q} $. The elements $ d_1, d_2, \ldots, d_r \in R $ are unique up to multiplication by units and permutation \cite{smith} and the diagonal matrix $ \mathbf{D} $ is called the \textit{Smith normal form} of $ \mathbf{A} $. Hence we may define ranks and free ranks as in \cite[Def. 3.3]{kamche}.

\begin{definition} 
Given $ \mathbf{A} \in R^{m \times n} $ with Smith normal form $ \mathbf{D} = \diag(d_1, d_2, \ldots, d_r) \in R^{m \times n} $, $ r = \min \{ m,n \} $, we define:
\begin{enumerate}
\item
The rank of $ \mathbf{A} $ as $ \rk(\mathbf{A}) = |\{ i \in [r] \mid d_i \neq 0 \}| $.
\item
The free rank of $ \mathbf{A} $ as $ \frk(\mathbf{A}) = |\{ i \in [r] \mid  d_i \in R^*  \}| $.
\end{enumerate}
\end{definition}

In this manuscript, we will mainly work with linear codes in $ S^n $. To that end, we will translate the rank metric from $ R^{m \times n} $ to $ S^n $ as in \cite[Sec. III-B]{kamche}. For a positive integer $ t $ and an ordered basis $ \boldsymbol\alpha = ( \alpha_1, \alpha_2, \ldots, \alpha_m ) \in S^m $ of $ S $ over $ R $, we define the matrix representation map $ M_{\boldsymbol\alpha} : S^t \longrightarrow R^{m \times t} $ by 
\begin{equation}
M_{\boldsymbol\alpha} \left( \mathbf{c} \right) = \left( \begin{array}{cccc}
c_{1,1} & c_{1,2} & \ldots & c_{1,t} \\
c_{2,1} & c_{2,2} & \ldots & c_{2,t} \\
\vdots & \vdots & \ddots & \vdots \\
c_{m,1} & c_{m,2} & \ldots & c_{m,t} \\
\end{array} \right) \in R^{m \times t},
\label{eq def matrix representation map}
\end{equation}
where $ \mathbf{c} = (c_1, c_2, \ldots, c_t) \in S^t $ and, for each $ j \in [t] $, $ c_{1,j}, c_{2,j}, \ldots, c_{m,j} \in R $ are the coordinates of $ c_j $ in the ordered basis $ \boldsymbol\alpha $, that is, they are the unique scalars in $ R $ such that $ c_j = \sum_{i=1}^m \alpha_i c_{i,j} $. Notice that also $ \mathbf{c} = \sum_{i=1}^m \alpha_i (c_{i,1}, c_{i,2}, \ldots, c_{i,t}) $. Using this matrix representation map, we define $ \rk(\mathbf{c}) = \rk \left( M_{\boldsymbol\alpha} \left( \mathbf{c} \right) \right) $ and $ \frk(\mathbf{c}) = \frk \left( M_{\boldsymbol\alpha} \left( \mathbf{c} \right) \right) $, which do not depend on the ordered basis $ \boldsymbol\alpha $, see also \cite{kamche}.

We may now define the sum-rank metric for the ring extension $ R \subseteq S $.
 
\begin{definition} [\textbf{Sum-rank metric}]
Consider positive integers $ n_1, n_2, \ldots, n_\ell $ and $ n = n_1 + n_2 + \cdots + n_\ell $. We define the sum-rank weight of $ \mathbf{c} \in S^n $ over $ R $ for the length partition $ n = n_1 + n_2 + \cdots + n_\ell $ as
$$ \wt_{SR}(\mathbf{c}) = \sum_{i=1}^\ell \rk \left( \mathbf{c}^{(i)} \right), $$
where $ \mathbf{c} = \left( \mathbf{c}^{(1)}, \mathbf{c}^{(2)}, \ldots, \mathbf{c}^{(\ell)} \right) $ and $ \mathbf{c}^{(i)} \in S^{n_i} $, for $ i \in [\ell] $. We define the sum-rank metric $ \dd_{SR} : S^{2n} \longrightarrow S^n $ over $ R $ for the length partition $ n = n_1 + n_2 + \cdots + n_\ell $ by 
$$ \dd_{SR}(\mathbf{c}, \mathbf{d}) = \wt_{SR}(\mathbf{c} - \mathbf{d}), $$
for $ \mathbf{c}, \mathbf{d} \in S^n $.
\end{definition}

This definition coincides with the classical one \cite{space-time-kumar, multishot} when $ R $ and $ S $ are fields. Over finite chain rings, this definition coincides with the Hamming metric when $ n_1 = n_2 = \ldots = n_\ell = 1 $ and with the rank metric as above \cite{kamche} when $ \ell = 1 $.

Once again, the definitions of the sum-rank weight and metric in $ S^n $ do not depend on the ordered basis $ \boldsymbol\alpha $. Furthermore, the sum-rank metric satisfies the properties of a metric since rank weights are norms by \cite[Th. 3.9]{kamche}. As noted in \cite[Remark 3.10]{kamche}, free ranks do not generally give rise to a metric nor include the Hamming metric over rings. The subring $ R $ and the length partition $ n = n_1 + n_2 + \cdots + n_\ell $ will not be specified unless necessary.

The following result will be crucial for our purposes. It can be proven as in \cite[Th. 1]{universal-lrc}, but using the Smith normal form.

\begin{lemma} \label{lemma SR is min of Hamming}
Given $ \mathbf{c} \in S^n $, and considering the subring $ R \subseteq S $ and the length partition $ n = n_1 + n_2 + \cdots + n_\ell $, it holds that
\begin{equation*}
\begin{split}
\wt_{SR}(\mathbf{c}) = \min \{ & \wt_H(\mathbf{c} \diag(\mathbf{A}_1, \mathbf{A}_2, \ldots, \mathbf{A}_\ell)) | \\
& \mathbf{A}_i \in R^{n_i \times n_i} \textrm{ invertible}, i \in [\ell] \} .
\end{split}
\end{equation*}
In particular, given an arbitrary code $ \mathcal{C} \subseteq S^n $ (linear or not),we have that
\begin{equation*}
\begin{split}
\dd_{SR}(\mathcal{C}) = \min \{ & \dd_H(\mathcal{C} \diag(\mathbf{A}_1, \mathbf{A}_2, \ldots, \mathbf{A}_\ell)) | \\
& \mathbf{A}_i \in R^{n_i \times n_i} \textrm{ invertible}, i \in [\ell] \} .
\end{split}
\end{equation*}
\end{lemma}

One immediate consequence of Lemma \ref{lemma SR is min of Hamming} above is the following classical version of the Singleton bound, but for the sum-rank metric for the ring extension $ R \subseteq S $. This bound recovers \cite[Prop. 34]{linearizedRS} when $ R $ and $ S $ are fields, and it recovers \cite[Prop. 3.20]{kamche} when $ \ell = 1 $.

\begin{proposition} [\textbf{Singleton bound}]
Given an arbitrary code $ \mathcal{C} \subseteq S^n $ (linear or not), and setting $ k = \log_{|S|}|\mathcal{C}| $, it holds that
$$ \dd_{SR}(\mathcal{C}) \leq n - k + 1. $$
\end{proposition}

We note that there exist more general Singleton bounds for the sum-rank metric over finite fields, see \cite[Th. III.2]{alberto-fundamental}. We leave as an open problem generalizing such bounds to finite chain rings. 

Thus we may define MSRD codes as follows. This definition recovers that of MSRD codes \cite[Th. 4]{linearizedRS} when $ R $ and $ S $ are fields, MDS codes over finite chain rings when $ n_1 = n_2 = \ldots = n_\ell = 1 $, and MRD codes over finite chain rings \cite[Def. 3.21]{kamche} when $ \ell = 1 $. 

\begin{definition} [\textbf{MSRD codes}] \label{def MSRD codes}
We say that a code $ \mathcal{C} \subseteq S^n $ is a \textit{maximum sum-rank distance} (MSRD) code over $ R $ for the length partition $ n = n_1 + n_2 + \cdots + n_\ell $ if $ k = \log_{|S|}|\mathcal{C}| $ is a positive integer and $ \dd_{SR}(\mathcal{C}) = n-k+1 $, where $ \dd_{SR} $ is considered for such a subring and length partition.
\end{definition}

From Lemma \ref{lemma SR is min of Hamming}, we deduce the following auxiliary lemma, which we will use in Section \ref{sec LRS} to prove that linearized Reed--Solomon codes are MSRD.

\begin{lemma} \label{lemma MSRD from MDS}
Given an arbitrary code $ \mathcal{C} \subseteq S^n $ (linear or not) such that $ k = \log_{|S|}|\mathcal{C}| $ is a positive integer, and for the subring $ R \subseteq S $ and length partition $ n = n_1 + n_2 + \cdots + n_\ell $, it holds that $ \mathcal{C} $ is MSRD if, and only if, the code $ \mathcal{C} \diag(\mathbf{A}_1, \mathbf{A}_2, \ldots, \mathbf{A}_\ell) $ is MDS, for all invertible matrices $ \mathbf{A}_i \in R^{n_i \times n_i} $, for $ i \in [\ell] $.
\end{lemma}

\section{Skew Polynomials on Finite Chain Rings} \label{sec skew polynomials}

We will extensively use \textit{skew polynomials} \cite{ore}, but defined over finite chain rings instead of fields or division rings. The ring of skew polynomials over $ S $ with morphism $ \sigma $ is the set $ S[x;\sigma] $ formed by elements of the form $ F = F_0 + F_1 x + F_2 x^2 + \cdots + F_d x^d $, for $ F_0, F_1, F_2, \ldots, F_d \in S $ and $ d \in \mathbb{N} $, as in the conventional polynomial ring. Furthermore, if $ F_d \neq 0 $, then we define the \textit{degree} of $ F $ as $ \deg(F) = d $, and we say that $ F $ is \textit{monic} if $ F_d = 1 $. If $ F = 0 $, then we define $ \deg(F) = - \infty $. Moreover, sums of skew polynomials and products with scalars on the left are defined as in the case of conventional polynomials. The only difference is that the product of skew polynomials is given by the rule
$$ x a = \sigma(a) x, $$
for $ a \in S $, together with the rule $ x^i x^j = x^{i+j} $, for $ i,j \in \mathbb{N} $.

In order to define linearized Reed--Solomon codes for the extension $ R \subseteq S $, we will need the following definitions. We start with the following operators, considered in \cite[Def.\ 3.1]{hilbert90} and \cite[Eq. (2.7)]{leroy-pol} for division rings. The definition can be trivially adapted to finite chain rings.

\begin{definition} [\textbf{\cite{hilbert90, leroy-pol}}] \label{def linearized operators}
Fix $ a \in S $ and define its $ i $th norm as $ N_i(a) = \sigma^{i-1}(a) \cdots \sigma(a)a $ for $ i \in \mathbb{N} $. Now define the $ R $-linear operator $ \mathcal{D}_a^i : S \longrightarrow S $ by
\begin{equation}
\mathcal{D}_a^i(\beta) = \sigma^i(\beta) N_i(a) ,
\label{eq definition linear operator}
\end{equation}
for all $ \beta \in S $, and all $ i \in \mathbb{N} $. Define also $ \mathcal{D}_a = \mathcal{D}_a^1 $ and observe that $ \mathcal{D}_a^{i+1} = \mathcal{D}_a \circ \mathcal{D}_a^i $, for $ i \in \mathbb{N} $. If $ \sigma $ is not understood from the context, we will write $ N^\sigma_i(a) $ and $ \mathcal{D}_{\sigma,a}^i(\beta) $, for $ i \in \mathbb{N} $, $ a, \beta \in S $.

Finally, given a skew polynomial $ F = \sum_{i=0}^d F_i x^i \in S[x;\sigma] $, where $ d \in \mathbb{N} $, we define its operator evaluation on the pair $ (a,\beta) \in S^2 $ as
$$ F_a(\beta) = \sum_{i=0}^d F_i \mathcal{D}_a^i (\beta) \in S. $$
Observe that $ F_a $ can be seen as an $ R $-linear map $ F_a : S \longrightarrow S $, taking $ \beta \in S $ to $ F_a(\beta) \in S $. 
\end{definition}

\begin{example} 
Let the setting be as in Example \ref{ex first}. Choose $ a = \alpha + 4 $ and $ \beta = \alpha $. Then
$$ \mathcal{D}_a^2(\beta) = \sigma^2(\beta) \sigma(a) a = \alpha^9 \cdot (\alpha^3+4) \cdot (\alpha+4) = 8 \alpha. $$
\end{example}

We will also need the concept of \textit{conjugacy}, introduced in \cite{lam, lam-leroy} for division rings. We adapt the definition to finite chain rings as follows.

\begin{definition} [\textbf{Conjugacy \cite{lam, lam-leroy}}] \label{def conjugacy}
Given $ a,b \in S $, we say that they are \textit{conjugate} in $ S $ with respect to $ \sigma $ if there exists $ \beta \in S^* $ such that $ b = a^\beta $, where
$$ a^\beta = \sigma(\beta) a \beta^{-1} . $$
\end{definition}

We now extend some results by Lam and Leroy \cite{lam, lam-leroy} to finite chain rings. These results will be crucial for defining and studying linearized Reed--Solomon codes.

The following result follows by combining \cite[Th. 23]{lam} and \cite[Th. 4.5]{lam-leroy}, and was presented in the following form in \cite[Th. 2.1]{leroy-noncommutative}. We only consider finite fields.

\begin{lemma} [\textbf{\cite{lam, lam-leroy}}] \label{lemma zeros lam leroy}
If $ \overline{a}_1, \overline{a}_2, \ldots, \overline{a}_\ell \in \mathbb{F}_{q^m}^* $ are pair-wise non-conjugate (with respect to $ \overline{\sigma} $) and $ F \in \mathbb{F}_{q^m}[x;\overline{\sigma}] $ is not zero, then
$$ \sum_{i=1}^\ell \dim_{\mathbb{F}_q}(\ker (F_{\overline{a}_i})) \leq \deg(F). $$
\end{lemma}

We now extend this result to the finite chain rings $ R \subseteq S $ (we will give a different extension in Lemma \ref{lemma degree bound with rank}). To this end, we define the \textit{free rank} of an $ R $-module $ M $ as the maximum size of an $ R $-linearly independent subset of $ M $. We will denote it by $ \frk_R(M) $.

\begin{theorem} \label{th zeros rings}
Let $ a_1, a_2, \ldots, a_\ell \in S^* $ be such that $ a_i - a_j^\beta \in S^* $, for all $ \beta \in S^* $, and for $ 1 \leq i < j \leq \ell $. For any non-zero monic $ F \in S[x;\sigma] $, we have 
$$ \sum_{i=1}^\ell \frk_R(F_{a_i}^{-1}(\mathfrak{M})) \leq \deg(F). $$
\end{theorem}
\begin{proof}
If $ F = F_0 + F_1 x + \cdots + F_d x^d $, where $ F_0, F_1, \ldots, F_d \in S $, denote $ \overline{F} = \overline{F}_0 + \overline{F}_1 x + \cdots + \overline{F}_d x^d \in \mathbb{F}_{q^m}[x;\overline{\sigma}] $. We have the following two facts:

1) We have that 
$$ \frk_R(F_a^{-1}(\mathfrak{M})) \leq \dim_{\mathbb{F}_q}(\ker(\overline{F}_{\overline{a}})). $$
We now prove this claim. From Definition \ref{def linearized operators} and the fact that $ \overline{\sigma(a)} = \overline{\sigma}(\overline{a}) $ (see (\ref{eq commutative diagram})), 
$$ \overline{F}_{\overline{a}}(\overline{\beta}) = \overline{F_a(\beta)}, $$
for all $ a, \beta \in S $. This means that, if $ F_a(\beta) \in \mathfrak{M} $, then $ \overline{F}_{\overline{a}}(\overline{\beta}) = \overline{F_a(\beta)} = 0 $. Therefore, $ \overline{F_a^{-1}(\mathfrak{M})} \subseteq \ker(\overline{F}_{\overline{a}}) $. By Item 2 in Lemma \ref{lemma linear independence}, $ \frk_R(F_a^{-1}(\mathfrak{M})) \leq \dim_{\mathbb{F}_q}(\overline{F_a^{-1}(\mathfrak{M})}) $. Thus we conclude that $ \frk_R(F_a^{-1}(\mathfrak{M})) \leq \dim_{\mathbb{F}_q}(\overline{F_a^{-1}(\mathfrak{M})}) \leq \dim_{\mathbb{F}_q}(\ker(\overline{F}_{\overline{a}})) $. 

2) For $ 1 \leq i < j \leq \ell $ and $ \overline{\beta} \in \mathbb{F}_{q^m}^* $, we have that $ \overline{a}_i \neq \overline{a}_j^{\overline{\beta}} $ since $ \beta \in S^* $ and $ a_i - a_j^\beta \notin \mathfrak{M} $. 

By 2), Lemma \ref{lemma zeros lam leroy} applies and, using 1), we conclude that
\begin{equation*}
\begin{split}
\sum_{i=1}^\ell \frk_R(F_{a_i}^{-1}(\mathfrak{M})) & \leq \sum_{i=1}^\ell \dim_{\mathbb{F}_q}(\ker(\overline{F}_{\overline{a}_i})) \\
& \leq \deg(\overline{F}) \\
& = \deg(F),
\end{split}
\end{equation*}
where $ \deg(F) = \deg(\overline{F}) $, since $ F $ is non-zero and monic. 
\end{proof}

Using Theorem \ref{th zeros rings}, we may prove the existence of monic annihilator skew polynomials and Lagrange interpolating skew polynomials of the smallest possible degree. To this end, we need more auxiliary tools. First, we need the following alternative notion of evaluation, introduced in \cite{lam, lam-leroy} for division rings and based on right Euclidean division \cite{ore}. The adaptation to finite chain rings is trivial.

\begin{definition} [\textbf{\cite{lam, lam-leroy}}] \label{def skew evaluation}
Given a skew polynomial $ F \in S[x;\sigma] $ and $ a \in S $, we define the remainder evaluation of $ F $ at $ a $, denoted by $ F(a) $, as the only scalar $ F(a) \in S $ such that there exists $ Q \in S[x;\sigma] $ with $ F = Q \cdot (x-a) + F(a) $.
\end{definition}

We will also need the product rule, given in \cite[Th. 2.7]{lam-leroy} for division rings, but which holds for finite chain rings as stated below.

\begin{lemma} [\textbf{\cite{lam-leroy}}] \label{lemma product rule}
Let $ F,G \in S[x;\sigma] $ and $ a \in S $. If $ G(a) = 0 $, then $ (FG) (a) = 0 $. If $ \beta = G(a) \in S^* $, then $ (FG)(a) = F(a^\beta)G(a) $.
\end{lemma}

Another tool that we will need is the following connection between the remainder evaluation as above and the evaluation from Definition \ref{def linearized operators}. It was proven in \cite[Lemma 1]{lam} for division rings, but it holds for finite chain rings as stated below.

\begin{lemma} [\textbf{\cite{lam}}] \label{lemma eval connection}
Given $ F \in S[x;\sigma] $, $ a \in S $ and $ \beta \in S^* $, it holds that
$$ F_a(\beta) = F ( a^\beta ) \beta . $$
\end{lemma}

We will show that annihilator skew polynomials and Lagrange interpolating skew polynomials exist for sequences of evaluation points as follows.

\begin{definition} \label{def P-independent}
Consider vectors $ \mathbf{a} = ( a_1 , a_2, \ldots, a_\ell) \in (S^*)^\ell $ and $ \boldsymbol\beta_i = (\beta_{i,1}, \beta_{i,2}, \ldots, $ $ \beta_{i,n_i}) $ $ \in S^{n_i} $, for $ i \in [\ell] $. Set $ \boldsymbol\beta = (\boldsymbol\beta_1, \boldsymbol\beta_2, \ldots, \boldsymbol\beta_\ell) $. We say that $ (\mathbf{a}, \boldsymbol\beta) $ satisfies the \textit{MSRD property} if the following conditions hold:
\begin{enumerate}
\item
$ a_i - a_j^\beta \in S^* $, for all $ \beta \in S^* $ and for $ 1 \leq i < j \leq \ell $.
\item
$ \beta_{i,1}, \beta_{i,2} , \ldots, \beta_{i,n_i} $ are linearly independent over $ R $, for $ i \in [\ell] $.
\end{enumerate}
Note that, by Item 3 in Lemma \ref{lemma linear independence}, $ \beta_{i,j} \in S^* $, for $ j \in [n_i] $ and $ i \in [\ell] $.
\end{definition}

\begin{example} \label{ex MSRD sequence} 
Let the setting be as in Example \ref{ex first}. Let $ a_1 = 1 $, $ a_2 = \alpha+1 $. Their images $ \overline{1}, \overline{\alpha}+\overline{1} $ in $ \mathbb{F}_9 $ satisfy 
$$ N_{\mathbb{F}_9/\mathbb{F}_3}(\overline{\alpha}+\overline{1}) = (\overline{\alpha}^3+\overline{1})(\overline{\alpha}+\overline{1}) = \overline{2} \neq \overline{1} = N_{\mathbb{F}_9/\mathbb{F}_3}(\overline{1}), $$
where $ N_{\mathbb{F}_9/\mathbb{F}_3} $ is the norm of the field extension $ \mathbb{F}_3 \subseteq \mathbb{F}_9 $. Thus by Hilbert's Theorem 90 it holds that $ a_1 - a_2^\beta \in S^* $, for all $ \beta \in S^* $. 

Finally choose $ \beta_1 = 1 $ and $ \beta_2 = \alpha $, which are clearly $ R $-linearly independent, and set $ \mathbf{a} = (a_1,a_2) $ and $ \boldsymbol\beta = (\beta_1, \beta_2) $. Therefore $ (\mathbf{a},\boldsymbol\beta) = ((1, \alpha+1),(1,\alpha)) $ satisfies the MSRD property.
\end{example}

The next step is the existence of minimal annihilator skew polynomials of the ``right'' degree. The following proposition recovers \cite[Prop. 2.5]{kamche} when $ \ell = 1 $ and $ a_1 = 1 $.

\begin{theorem} \label{th central for min skew pol}
Let $ (\mathbf{a},\boldsymbol\beta) $ be as in Definition \ref{def P-independent}, and satisfying the MSRD property. Then there exist units $ \gamma_{i,j} \in S^* $, with $ \gamma_{1,1} = \beta_{1,1} $, and skew polynomials of the form
\begin{equation*}
\begin{split}
G_{i,j} = & \left( x - a_i^{\gamma_{i,j}} \right) \cdots \left( x - a_i^{\gamma_{i,1}} \right) \cdot \\
& \left( x - a_{i-1}^{\gamma_{i-1,n_{i-1}}} \right) \cdots \left( x - a_{i-1}^{\gamma_{i-1,1}} \right) \cdots \\
& \left( x - a_1^{\gamma_{1,n_1}} \right) \cdots \left( x - a_1^{\gamma_{1,1}} \right) \in S[x;\sigma] ,
\end{split}
\end{equation*}
of degree $ \deg(G_{i,j}) = \sum_{u=1}^{i-1} n_u + j $, and such that
\begin{equation*}
\begin{split}
G_{i,j}(a_u^{\beta_{u,v}}) = 0, & \textrm{ if } 1 \leq u \leq i-1, \\
& \textrm{ or if } u=i \textrm{ and } 1 \leq v \leq j, \\
G_{i,j}(a_u^{\beta_{u,v}}) \in S^*, & \textrm{ if } i+1 \leq u \leq \ell, \\
& \textrm{ or if } u=i \textrm{ and } j+1 \leq v \leq n_i,
\end{split}
\end{equation*}
and $ G_{i,j} $ is unique among monic skew polynomials in $ S[x;\sigma] $ satisfying such properties, for $ j \in [n_i] $ and $ i \in [\ell] $.
\end{theorem}
\begin{proof}
We prove the proposition by induction in the pair $ (i,j) $. For the basis step, we only need to define $ G_{1,1} = x - a_1^{\beta_{1,1}} $. We have $ G_{1,1,a_1}(\beta_{1,1}) = 0 $ by Lemma \ref{lemma eval connection}. On the other hand, since $ \deg(G_{1,1}) = 1 $ and it is non-zero and monic, then $ G_{1,1,a_u}(\beta_{u,v}) \in S^* $, if $ (u,v) \neq (1,1) $, by Theorem \ref{th zeros rings} and Lemma \ref{lemma eval connection}.

Now, we have two cases for the inductive step. Either we go from $ G_{i,j} $ to $ G_{i,j+1} $, if $ j < n_i $, or from $ G_{i,n_i} $ to $ G_{i+1,1} $ if $ i < \ell $. The process stops when $ i=\ell $ and $ j=n_\ell $. We will only develop the first case of induction step, since the second case is analogous.

Assume that $ G_{i,j} $ satisfies the properties in the proposition and $ j < n_i $. In particular, $ G_{i,j}(a_i^{\beta_{i,j+1}}) \in S^* $. Thus, we may define $ \gamma_{i,j+1} = G_{i,j}(a_i^{\beta_{i,j+1}}) \beta_{i,j+1} \in S^* $ and
$$ G_{i,j+1} = \left( x - a_i^{\gamma_{i,j+1}} \right) G_{i,j}. $$
By Lemmas \ref{lemma product rule} and \ref{lemma eval connection} and the assumptions on $ G_{i,j} $, we have that $ G_{i,j+1}(a_u^{\beta_{u,v}}) = 0 $, if $ 1 \leq u \leq i-1 $, or if $ u=i $ and $ 1 \leq v \leq j+1 $. Since $ G_{i,j+1} $ has such a set of zeros, it is non-zero, monic and of degree $ \sum_{u=1}^{i-1} n_u + j+1 $, then we deduce from Theorem \ref{th zeros rings} and Lemma \ref{lemma eval connection} that $ G_{i,j+1}(a_u^{\beta_{u,v}}) \in S^* $, if $ i+1 \leq u \leq \ell $, or if $ u=i $ and $ j+2 \leq v \leq n_i $.

Finally, the uniqueness of $ G_{i,j} $ follows by combining Theorem \ref{th zeros rings} and Lemma \ref{lemma eval connection}.
\end{proof}

We immediately deduce the following two consequences. The first of these corollaries is the existence of annihilator skew polynomials of minimum possible degree.

\begin{corollary} \label{cor annihilator skew pol}
Let $ (\mathbf{a},\boldsymbol\beta) $ be as in Definition \ref{def P-independent}, and satisfying the MSRD property. Then there exists a unique monic skew polynomial $ F \in S[x;\sigma] $ such that $ \deg(F) = n_1 + n_2 + \cdots + n_\ell $ and $ F_{a_i}(\beta_{i,j}) = 0 $, for $ j \in [n_i] $ and $ i \in [\ell] $.
\end{corollary}
\begin{proof}
Take $ F = G_{\ell, n_\ell} $ in Theorem \ref{th central for min skew pol}.
\end{proof}

The second corollary states the existence of a basis for Lagrange interpolation.

\begin{corollary} \label{cor dual basis}
Let $ (\mathbf{a},\boldsymbol\beta) $ be as in Definition \ref{def P-independent}, and satisfying the MSRD property. For each $ j \in [n_i] $ and $ i \in [\ell] $, there exists a unique skew polynomial $ F_{i,j} \in S[x;\sigma] $ such that $ \deg(F_{i,j}) = n_1 + n_2 + \cdots + n_\ell - 1 $, $ F_{i,j,a_i}(\beta_{i,j}) = 1 $, and $ F_{i,j,a_u}(\beta_{u,v}) = 0 $, for all $ v \in [n_i] $ and $ u \in [\ell] $ with $ u \neq i $ or $ v \neq j $.
\end{corollary}
\begin{proof}
Up to reordering, we may assume that $ i = \ell $ and $ j = n_\ell $. With notation as in Theorem \ref{th central for min skew pol}, let $ G = G_{\ell,n_\ell-1} $ if $ n_\ell > 1 $, or $ G = G_{\ell-1,n_{\ell -1}} $ if $ n_\ell = 1 $. By Lemma \ref{lemma eval connection}, since $ G(a_\ell^{\beta_{\ell,n_\ell}}) \in S^* $ and $ \beta_{\ell,n_\ell} \in S^* $, then $ G_{a_\ell}(\beta_{\ell,n_\ell}) \in S^* $. Hence, we are done by defining $ F_{\ell,n_\ell} = G_{a_\ell}(\beta_{\ell,n_\ell})^{-1} G $. The uniqueness follows again from Theorem \ref{th zeros rings}.
\end{proof}

We may also obtain the following strengthening of Corollary \ref{cor annihilator skew pol} on monic annihilator skew polynomials. It is a generalization of \cite[Prop. 3.15]{kamche}.

\begin{corollary} \label{cor annihilator skew pol from rank}
Let $ a_1, a_2, \ldots, a_\ell \in S $ be such that $ a_i - a_j^\beta \in S^* $, for all $ \beta \in S^* $ and for $ 1 \leq i < j \leq \ell $. Let $ \mathbf{u}_i \in S^{n_i} $ and let $ t_i = \rk (\mathbf{u}_i) $, for $ i \in [\ell] $. Set $ t = t_1 + t_2 + \cdots + t_\ell $. Then there exists a monic skew polynomial $ F \in S[x;\sigma] $ such that $ \deg(F) = t $ and $ F_{a_i}(u_{i,j}) = 0 $, for $ j \in [n_i] $ and for $ i \in [\ell] $.
\end{corollary}
\begin{proof}
Using the Smith normal form (see Section \ref{sec MSRD codes}), we see that there are $ \boldsymbol\alpha_i \in S^{t_i} $ and $ \mathbf{B}_i \in R^{t_i \times n_i} $ such that $ \mathbf{u}_i = \boldsymbol\alpha_i \mathbf{B}_i $, $ \frk(\boldsymbol\alpha_i) = t_i $ and $ \rk (\mathbf{B}_i) = t_i $, for $ i \in [\ell] $. In particular, $ (\mathbf{a}, \boldsymbol\alpha) $ satisfies the MSRD property (Definition \ref{def P-independent}), where $ \mathbf{a} = (a_1, a_2, \ldots, a_\ell) $ and $ \boldsymbol\alpha = (\boldsymbol\alpha_1, \boldsymbol\alpha_2, \ldots, \boldsymbol\alpha_\ell) $. By Corollary \ref{cor annihilator skew pol}, there exists a monic skew polynomial $ F \in S[x;\sigma] $ such that $ \deg(F) = t $ and $ F_{a_i}(\alpha_{i,j}) = 0 $, for $ j \in [t_i] $ and for $ i \in [\ell] $. Since the map $ F_{a_i} $ is $ R $-linear and $ \mathbf{u}_i = \boldsymbol\alpha_i \mathbf{B}_i $, we deduce that $ F_{a_i}(u_{i,j}) = 0 $, for $ j \in [n_i] $ and for $ i \in [\ell] $, and we are done.
\end{proof}

Next we define extended Moore matrices for the ring extension $ R \subseteq S $. Such matrices are a trivial adaptation of the matrices from \cite[p. 604]{linearizedRS} from division rings to finite chain rings. These matrices will be used to define linearized Reed--Solomon codes and to explore further forms of Lagrange interpolation.

\begin{definition} \label{def extend Moore}
Consider vectors $ \mathbf{a} = ( a_1 , a_2, \ldots, a_\ell) \in S^\ell $ and $ \boldsymbol\beta_i = (\beta_{i,1}, \beta_{i,2}, \ldots, \beta_{i,n_i}) \in S^{n_i} $, for $ i \in [\ell] $. Set $ \boldsymbol\beta = (\boldsymbol\beta_1, \boldsymbol\beta_2, \ldots, \boldsymbol\beta_\ell) $ and $ n = n_1 + n_2 + \cdots + n_\ell $. For $ k \in [n] $, we define the \textit{extended Moore matrix} $ \mathbf{M}_k(\mathbf{a}, \boldsymbol\beta) = $
\begin{equation*}
\resizebox{\columnwidth}{!}{%
$
\left( \begin{array}{ccc|c|ccc}
\beta_{1,1} & \ldots & \beta_{1, n_1} & \ldots & \beta_{\ell, 1} & \ldots & \beta_{\ell, n_\ell} \\
\mathcal{D}_{a_1}(\beta_{1,1}) & \ldots & \mathcal{D}_{a_1}(\beta_{1,n_1}) & \ldots & \mathcal{D}_{a_\ell}(\beta_{\ell,1}) & \ldots & \mathcal{D}_{a_\ell}(\beta_{\ell, n_\ell}) \\
\mathcal{D}_{a_1}^2(\beta_{1,1}) & \ldots & \mathcal{D}_{a_1}^2(\beta_{1,n_1}) & \ldots & \mathcal{D}_{a_\ell}^2(\beta_{\ell,1}) & \ldots & \mathcal{D}_{a_\ell}^2(\beta_{\ell,n_\ell}) \\
\vdots & \ddots & \vdots & \ddots & \vdots & \ddots & \vdots \\
\mathcal{D}_{a_1}^{k-1}(\beta_{1,1}) & \ldots & \mathcal{D}_{a_1}^{k-1}(\beta_{1,n_1}) & \ldots &  \mathcal{D}_{a_\ell}^{k-1}(\beta_{\ell,1}) & \ldots & \mathcal{D}_{a_\ell}^{k-1}(\beta_{\ell,n_\ell}) \\
\end{array} \right) .
$
}%
\end{equation*}
When there is confusion about $ \sigma $, we will write  $ \mathbf{M}^\sigma_k(\mathbf{a},\boldsymbol\beta) $ instead of $ \mathbf{M}_k(\mathbf{a},\boldsymbol\beta) $.
\end{definition}

\begin{example}  
Let the setting be as in Example \ref{ex first}. Let $ \mathbf{a} = (a_1,a_2) = (1, \alpha+1) $ and $ \boldsymbol\beta = (\beta_1, \beta_2) = (1,\alpha) $. In Example \ref{ex MSRD sequence}, we saw that $ (\mathbf{a},\boldsymbol\beta) = ((1, \alpha+1),(1,\alpha)) $ satisfies the MSRD property. Notice that $ \ell = m = n_1 = n_2 = 2 $. If we set $ k=3 $, then the corresponding extended Moore matrix is
$$ \mathbf{M}_3(\mathbf{a},\boldsymbol\beta) = \left( \begin{array}{cc|cc}
1 & \alpha & 1 & \alpha \\
1 & -\alpha & \alpha+1 & 8\alpha +1 \\
1 & \alpha & 2 & 2 \alpha \\
\end{array} \right). $$
\end{example}

The following result gives a sufficient condition for extended Moore matrices over finite chain rings to be invertible, and it may be of interest on its own.

\begin{theorem} \label{th ext moore matrix invertible}
Let $ (\mathbf{a},\boldsymbol\beta) $ be as in Definition \ref{def P-independent}, and satisfying the MSRD property. Let $ n = n_1 + n_2 + \cdots + n_\ell $. Then the square extended Moore matrix $ \mathbf{M}_n(\mathbf{a},\boldsymbol\beta) $ is invertible. 
\end{theorem}
\begin{proof}
Let $ F_{i,j} \in S[x;\sigma] $ be as in Corollary \ref{cor dual basis}, for $ j \in [n_i] $ and $ i \in [\ell] $. Then, for the appropriate ordering, the coefficients of such skew polynomials (they are of degree $ n-1 $) form the rows of the inverse of $ \mathbf{M}_n(\mathbf{a},\boldsymbol\beta) $.
\end{proof}

From Theorem \ref{th ext moore matrix invertible}, we may obtain the following Lagrange interpolation theorem, which we will use later for decoding and may be of interest on its own.

\begin{theorem} \label{th lagrange interpolation}
Let $ (\mathbf{a},\boldsymbol\beta) $ be as in Definition \ref{def P-independent}, and satisfying the MSRD property. Let $ c_{i,j} \in S $, for $ j \in [n_i] $ and $ i \in [\ell] $. Then there exists a unique skew polynomial $ F \in S[x;\sigma] $ such that $ \deg(F) \leq n_1 + n_2 + \cdots + n_\ell - 1 $, and $ F_{a_i}(\beta_{i,j}) = c_{i,j} $, for $ j \in [n_i] $ and $ i \in [\ell] $.
\end{theorem}

\section{Linearized Reed--Solomon Codes} \label{sec LRS}

In this section, we extend the definition of linearized Reed--Solomon codes \cite{linearizedRS} to finite chain rings, thus providing a first explicit construction of MSRD codes over finite chain rings (that are not fields).

\begin{definition} \label{def LRS codes}
Let $ (\mathbf{a},\boldsymbol\beta) $ be as in Definition \ref{def P-independent}, and satisfying the MSRD property. For $ k \in [n] $, we define the $ k $-dimensional linearized Reed--Solomon code as the linear code $ \mathcal{C}_k(\mathbf{a},\boldsymbol\beta) \subseteq S^n $ with generator matrix $ \mathbf{M}_k(\mathbf{a},\boldsymbol\beta) $ as in Definition \ref{def extend Moore}. When there is confusion about $ \sigma $, we will write  $ \mathcal{C}^\sigma_k(\mathbf{a},\boldsymbol\beta) $ instead of $ \mathcal{C}_k(\mathbf{a},\boldsymbol\beta) $.
\end{definition}

This definition coincides with \cite[Def. 31]{linearizedRS} when $ R $ and $ S $ are fields. It coincides with Gabidulin codes over finite chain rings \cite[Def. 3.22]{kamche} when $ \ell = 1 $ and generalized Reed--Solomon codes over finite chain rings \cite[Def. 22]{quintin} when $ m = n_1 = n_2 = \ldots = n_\ell = 1 $. 

The main result of this section is the following. 

\begin{theorem} \label{th LRS are MSRD}
Let $ (\mathbf{a},\boldsymbol\beta) $ be as in Definition \ref{def P-independent}, and satisfying the MSRD property. For $ k \in [n] $, the linearized Reed--Solomon code $ \mathcal{C}_k(\mathbf{a},\boldsymbol\beta) \subseteq S^n $ is a free $ S $-module of rank $ k $ and an MSRD code over $ R $ for the length partition $ n = n_1 + n_2 + \cdots + n_\ell $.
\end{theorem}
\begin{proof}
Let $ \mathbf{A}_i \in R^{n_i \times n_i} $ be invertible, for $ i \in [\ell] $. By the $ R $-linearity of $ \sigma $, we have that
$$ \mathcal{C}_k(\mathbf{a},\boldsymbol\beta) \diag(\mathbf{A}_1, \mathbf{A}_2, \ldots, \mathbf{A}_\ell ) $$
$$ = \mathcal{C}_k (\mathbf{a}, \boldsymbol\beta \diag(\mathbf{A}_1, \mathbf{A}_2, \ldots, \mathbf{A}_\ell )), $$
which is also a linearized Reed--Solomon code, since $ (\mathbf{a}, \boldsymbol\beta \diag(\mathbf{A}_1, \mathbf{A}_2, \ldots, \mathbf{A}_\ell )) $ also satisfies the MSRD property since $ \mathbf{A}_1, \mathbf{A}_2, \ldots, \mathbf{A}_\ell $ are invertible. Therefore, from Lemma \ref{lemma MSRD from MDS}, we see that we only need to prove that $ \mathcal{C}_k(\mathbf{a},\boldsymbol\beta) $ is MDS and a free $ S $-module of rank $ k $. Both properties follow from the fact that any $ k \times k $ square submatrix of $ \mathbf{M}_k(\mathbf{a},\boldsymbol\beta) $ is invertible by Theorem \ref{th ext moore matrix invertible}. 
\end{proof}

This result coincides with \cite[Th. 4]{linearizedRS} when $ R $ and $ S $ are fields, with \cite[Th. 3.24]{kamche} over finite chain rings when $ \ell = 1 $, and with \cite[Prop. 23 \& Cor. 24]{quintin} over finite chain rings when $ m = n_1 = n_2 = \ldots = n_\ell = 1 $.

Next, we show how to explicitly construct sequences $ (\mathbf{a},\boldsymbol\beta) $ satisfying the MSRD property. In this way, we have explicitly constructed linearized Reed--Solomon codes for the finite chain ring extension $ R \subseteq S $.

The $ R $-linearly independent elements $ \beta_{i,1}, \beta_{i,2}, \ldots, \beta_{i,n_i} \in S^* $ can be chosen as subsets of any basis of $ S $ over $ R $, for $ i \in [\ell] $. The more delicate part is choosing the elements $ a_1, a_2, \ldots, a_\ell \in S $. We now show two ways to do this. The proof of the following proposition is straightforward.

\begin{proposition}
Let $ \ell \in [q-1] $ and let $ \gamma \in \mathbb{F}_{q^m}^* $ be a primitive element, that is, $ \mathbb{F}_{q^m}^* = \{ \gamma^0, \gamma^1, \ldots, \gamma^{q^m-2} \} $. Such an element always exists \cite[Th. 2.8]{lidl}. Take elements $ a_1, a_2, \ldots, a_\ell \in S^* $ such that $ \overline{a}_i = \gamma^{i-1} $, for $ i \in [\ell] $. Then $ a_1, a_2, \ldots, a_\ell \in S^* $ are such that $ a_i - a_j^\beta \in S^* $, for all $ \beta \in S^* $ and all $ 1 \leq i < j \leq \ell $.
\end{proposition}

Another possibility is to choose elements from $ R^* $ when $ q-1 $ and $ m $ are coprime.

\begin{proposition} \label{prop conjugacy repr in base field}
Assume that $ q-1 $ and $ m $ are coprime and let $ \ell \in [q-1] $. Given $ a_1, a_2, \ldots, a_\ell \in R^* $, it holds that $ a_i - a_j \in R^* $ for all $ 1 \leq i < j \leq \ell $ if, and only if, $ \overline{a}_1, \overline{a}_2 ,\ldots, \overline{a}_\ell \in \mathbb{F}_q^* $ are all distinct. Moreover, if that is the case, then $ a_i - a_j^\beta \in S^* $, for all $ \beta \in S^* $ and all $ 1 \leq i < j \leq \ell $.
\end{proposition}
\begin{proof}
The first part is trivial, since $ R^* = R \setminus \mathfrak{m} $ and $ \mathfrak{m} = \ker (\rho) $. 
Now, since $ q-1 $ and $ m $ are coprime, it follows from \cite[Lemma 26]{SR-BCH} that $ \overline{a}_1, \overline{a}_2, \ldots, \overline{a}_\ell \in \mathbb{F}_q^* $ are pair-wise non-conjugate. Therefore, $ a_i - a_j^\beta \in S^* $, for all $ \beta \in S^* $ and all $ 1 \leq i < j \leq \ell $.
\end{proof}

Observe that in the previous two propositions, the maximum length of the vector $ (a_1, a_2, \ldots, a_\ell) $ is $ \ell = q-1 $. In the next proposition, we show that this is indeed the maximum possible.

\begin{proposition}
Let $ (a_1, a_2, \ldots , a_\ell) \in (S^*)^\ell $ be such that $ a_i - a_j^\beta \in S^* $, for all $ \beta \in S^* $ and all $ 1 \leq i < j \leq \ell $. Then $ \ell \leq q-1 $.
\end{proposition}
\begin{proof}
By the hypothesis on $ a_1, a_2, \ldots, a_\ell $, we deduce that $ \overline{a}_1, \overline{a}_2, \ldots, \overline{a}_\ell \in \mathbb{F}_{q^m}^* $ are pair-wise non-conjugate. Now, as shown in \cite{matroidal} (see also \cite[Prop. 45]{linearizedRS}), there are at most $ q-1 $ non-zero conjugacy classes in $ \mathbb{F}_{q^m} $ with respect to $ \overline{\sigma} $, that is, $ \ell \leq q-1 $.
\end{proof}

In particular, we have shown the existence of linear MSRD codes of any rank for the extension $ R \subseteq S $ as detailed in the following corollary. 

\begin{corollary} 
Let $ \ell \in [q-1] $, $ n_i \in [m] $ for $ i \in [\ell] $, and let $ k \in [n] $, where $ n = n_1 + n_2 + \cdots + n_\ell $. Then there exists a linear code $ \mathcal{C} \subseteq S^n $ that is a free $ S $-module of rank $ k $ and is MSRD over $ R $ for the length partition $ n = n_1 + n_2 + \cdots + n_\ell $.
\end{corollary}

For linearized Reed--Solomon codes, notice that these are the same parameter restrictions as in the finite-field case \cite[Sec. 4.2]{linearizedRS}.

We observer that, in the case of finite fields and square matrices ($ m = n_1 = n_2 = \ldots = n_\ell $), we have the upper bound $ \ell \leq q + \left\lfloor \frac{d-3}{n} \right\rfloor $ \cite[Th. VI.12]{alberto-fundamental}. This bound might hold also for finite chain rings, but we leave it as an open problem. Furthermore, in the Hamming-metric case ($ m=n_1 = \ldots = n_\ell = 1 $) it is conjectured that $ \ell \leq q+1 $ in general. Hence being able to attain the number of blocks $ \ell = q-1 $ is close to the known upper bounds on $ \ell $ for the case of finite fields. In the non-square case ($ m > n_i $), one may construct MSRD codes with an unrestricted number of blocks, see \cite[Subsec. 4.5]{generalMSRD}. 

Finally, we show that duals of linearized Reed--Solomon codes are again linearized Reed--Solomon codes. For a linear code $ \mathcal{C} \subseteq S^n $, we define its dual as $ \mathcal{C}^\perp = \{ \mathbf{d} \in S^n \mid \mathbf{c} \cdot \mathbf{d} = 0 \} $, where $ \cdot $ denotes the usual Euclidean inner product in $ S^n $. The following lemma follows from \cite[Th. 3.1]{honold}.

\begin{lemma} [\textbf{\cite{honold}}]
Given a linear code $ \mathcal{C} \subseteq S^n $, we have that $ \mathcal{C}^\perp $ is a free module if and only if, so is $ \mathcal{C} $. In such a case, if $ \mathcal{C} $ is of rank $ k $, then $ \mathcal{C}^\perp $ is of rank $ n-k $. Furthermore, $ \mathcal{C}^{\perp \perp} = \mathcal{C} $.
\end{lemma}

Using this lemma, we may prove that the dual of a linearized Reed--Solomon code is again a linearized Reed--Solomon code in the same way as in \cite[Th. 4]{secure-multishot}.

\begin{theorem} \label{th dual of LRS}
Let $ (\mathbf{a},\boldsymbol\beta) $ be as in Definition \ref{def P-independent}, and satisfying the MSRD property. There exists a vector $ \boldsymbol\delta = \left(\boldsymbol\delta^{(1)}, \boldsymbol\delta^{(2)}, \ldots, \boldsymbol\delta^{(\ell)} \right) \in S^n $, where $ \boldsymbol\delta^{(i)} = \left( \delta_1^{(i)}, \delta_2^{(i)}, \ldots, \right. $ $ \left. \delta_{n_i}^{(i)} \right) \in S^{n_i} $ and $ \delta_1^{(i)}, \delta_2^{(i)}, \ldots, $ $ \delta_{n_i}^{(i)} $ are $ R $-linearly independent, for $ i \in  [\ell] $, and such that 
\begin{equation}
\mathcal{C}^\sigma_k (\mathbf{a}, \boldsymbol\beta)^\perp = \mathcal{C}^{\sigma^{-1}}_{n-k} \left(\sigma^{-1} (\mathbf{a}), \boldsymbol\delta \right),
\label{eq dual of lin RS is lin RS}
\end{equation}
for $ k  \in  [n-1] $, where $ \sigma^{-1}(\mathbf{a}) = (\sigma^{-1}(a_1), \sigma^{-1}(a_2), \ldots, \sigma^{-1}(a_\ell)) $. Notice that $ (\sigma^{-1} (\mathbf{a}), \boldsymbol\delta) $ also satisfies the MSRD property.

Furthermore, if $ q-1 $ and $ m $ are coprime, and $ a_1, a_2, \ldots, a_\ell \in R^* $ are such that $ a_i - a_j \in R^* $ for all $ 1 \leq i < j \leq \ell $ (see Proposition \ref{prop conjugacy repr in base field}), then
\begin{equation}
\mathcal{C}^\sigma_k (\mathbf{a}, \boldsymbol\beta)^\perp = \mathcal{C}^{\sigma^{-1}}_{n-k} \left(\mathbf{a}, \boldsymbol\delta \right).
\label{eq dual of lin RS is lin RS base field}
\end{equation}
\end{theorem}

We will use the form of the dual of a linearized Reed--Solomon code shown in (\ref{eq dual of lin RS is lin RS base field}) to describe a quadratic-time decoding algorithm in Section \ref{sec syndrome decoder}.

\section{A Welch-Berlekamp Decoder} \label{sec WB decoder}

In this section, we present a Welch-Berlekamp sum-rank error-correcting algorithm for the linearized Reed--Solomon codes from Definition \ref{def LRS codes}. The decoder is based on the original one by Welch and Berlekamp \cite{welch-berlekamp}. Welch-Berlekamp decoders for the sum-rank metric in the case of fields were given in \cite{boucher-skew, secure-multishot, sven-decoder}, listed in decreasing order of computational complexity. Our decoder has cubic complexity over the ring $ S $ and is analogous to the works listed above. In Section \ref{sec syndrome decoder}, we will present a decoder with quadratic complexity, but which only works if $ q-1 $ and $ m $ are coprime. The decoder in this section works for all cases.

Throughout this section, we fix $ (\mathbf{a}, \boldsymbol\beta) $ as in Definition \ref{def P-independent}, and satisfying the MSRD property. Let
$$  b_{i,j} = a_i^{\beta_{i,j}},  $$
for $ j \in [n_i] $ and for $ i \in [\ell] $. Next fix a dimension $ k \in [n-1] $, and consider the linearized Reed--Solomon code $ \mathcal{C}_k (\mathbf{a}, \boldsymbol\beta) \subseteq S^n $ (Definition \ref{def LRS codes}). The number of sum-rank errors that it can correct is
\begin{equation}
t = \left\lfloor \frac{ \dd_{SR} \left( \mathcal{C}_k (\mathbf{a}, \boldsymbol\beta) \right) - 1 }{2} \right\rfloor = \left\lfloor
\frac{n-k}{2} \right\rfloor .
\label{eq ch 2 def t for welch-berlekamp}
\end{equation}
Let $ \mathbf{c} \in \mathcal{C}_k (\mathbf{a}, \boldsymbol\beta) $ be any \textit{codeword}, let $ \mathbf{e} \in S^n $ be an \textit{error vector} such that $ \wt_{SR} (\mathbf{e}) \leq t $, and define the \textit{received word} as
\begin{equation}
\mathbf{r} = \mathbf{c} + \mathbf{e} \in S^n .
\label{eq ch 2 def received word}
\end{equation}
Since $ \wt_{SR}(\mathbf{e}) \leq t $ and $ 2t + 1 \leq \dd_{SR} \left( \mathcal{C}_k (\mathbf{a}, \boldsymbol\beta) \right) $, there is a unique solution $ \mathbf{c} \in \mathcal{C}_k (\mathbf{a}, \boldsymbol\beta) $ to the decoding problem.

We start by defining the auxiliary vectors
\begin{equation}
\begin{split}
\mathbf{c}^\prime & = \mathbf{c} \cdot \diag(\boldsymbol\beta)^{-1}, \\
\mathbf{e}^\prime & = \mathbf{e} \cdot \diag(\boldsymbol\beta)^{-1}, \textrm{ and} \\ \mathbf{r}^\prime & = \mathbf{r} \cdot \diag(\boldsymbol\beta)^{-1}. 
\end{split}
\label{eq ch 2 def skew versions vectors decoding}
\end{equation}
By Lagrange interpolation (Theorem~\ref{th lagrange interpolation}) and Lemma \ref{lemma eval connection}, there exist unique skew polynomials $ F ,G, R \in S[x;\sigma] $, all of degree less than $ n $, such that
\begin{equation}
F(\mathbf{b}) = \mathbf{c}^\prime, \quad G(\mathbf{b}) = \mathbf{e}^\prime, \quad \textrm{and} \quad R(\mathbf{b}) = \mathbf{r}^\prime ,
\label{eq ch 2 skew pol versions of vectors decoding}
\end{equation}
which denote component-wise remainder evaluation (Definition \ref{def skew evaluation}). Following the original idea of the Welch--Berlekamp decoding algorithm, we want to find a non-zero monic skew polynomial $ L \in S[x;\sigma] $ with $ \deg(L) \leq t $ and such that
\begin{equation}
(LR)(\mathbf{b}) = (LF)(\mathbf{b}).
\label{eq ch 2 decoder key equation 1}
\end{equation}
However, since we do not know $ F $, we look instead for non-zero $ L,Q \in S[x;\sigma] $ such that $ L $ is monic, $ \deg(L) \leq t $, $ \deg(Q) \leq t + k - 1 $ and
\begin{equation}
(LR)(\mathbf{b}) = Q(\mathbf{b}).
\label{eq ch 2 decoder key equation 2}
\end{equation}

In the following two lemmas, we show that (\ref{eq ch 2 decoder key equation 1}) and (\ref{eq ch 2 decoder key equation 2}) can be solved, and once $ L $ and $ Q $ are obtained, $ F $ may be obtained in quadratic time (by Euclidean division).

\begin{lemma} \label{lemma ch 2 decoder lemma 1}
There exists a non-zero monic skew polynomial $ L \in S[x;\sigma] $ with $ \deg(L) \leq t $ satisfying (\ref{eq ch 2 decoder key equation 1}). In particular, there exist non-zero $ L,Q \in S[x;\sigma] $ such that $ L $ is monic, $ \deg(L) \leq t $, $ \deg(Q) \leq t + k - 1 $ and (\ref{eq ch 2 decoder key equation 2}) holds.
\end{lemma}
\begin{proof}
By Corollary \ref{cor annihilator skew pol from rank}, there exists a non-zero monic skew polynomial $ L \in S[x;\sigma] $ such that $ \deg(L) \leq t $ and $ L_{a_i}(e_{i,j}) = 0 $, for $ j \in [n_i] $ and for $ i \in [\ell] $. From the definitions and Lemma \ref{lemma eval connection}, it follows that 
$$ (LG)(b_{i,j}) = L_{b_{i,j}}(G(b_{i,j})) = L_{b_{i,j}}(e_{i,j}^\prime) = L_{a_i}(e_{i,j}) = 0, $$
for $ j \in [n_i] $ and for $ i \in [\ell] $. Since $ R(\mathbf{b}) = F(\mathbf{b}) + G(\mathbf{b}) $, we conclude that
$$ (L(R - F))(\mathbf{b}) = (LG)(\mathbf{b}) = 0 $$
by Lemma \ref{lemma product rule}. In other words, $ L $ satisfies (\ref{eq ch 2 decoder key equation 1}) and we are done.
\end{proof}

\begin{lemma} \label{lemma ch 2 decoder lemma 2}
If $ L, Q \in S[x;\sigma] $ are such that $ L $ is monic, $ \deg(L) \leq t $, $ \deg(Q) \leq t + k - 1 $ and (\ref{eq ch 2 decoder key equation 2}) holds, then
$$ Q = LF. $$
\end{lemma}
\begin{proof}
First, by (\ref{eq ch 2 decoder key equation 2}) and the product rule (Lemma \ref{lemma product rule}), 
$$ \textrm{if} \quad (F-R)(b_{i,j}) = 0, \quad \textrm{then} \quad (LF - Q)(b_{i,j}) = 0, $$
for $ j \in [n_i] $ and for $ i \in [\ell] $. From this fact, and using Lemmas \ref{lemma SR is min of Hamming} and \ref{lemma eval connection}, the
reader may deduce that
$$ \wt_{SR} \left( (LF - Q)(\mathbf{b}) \cdot \diag(\boldsymbol\beta) \right) $$
$$ \leq \wt_{SR} \left( (F-R)(\mathbf{b}) \cdot \diag(\boldsymbol\beta) \right) \leq t. $$
Therefore, we may apply Lemma~\ref{lemma ch 2 decoder lemma 1} to $ LF $ and $ Q $, instead of $ F $ and $ R $. Thus there exists a non-zero monic $ L_0 \in S[x;\sigma] $ such that $ \deg(L_0) \leq t $ and 
$$ (L_0 (LF - Q))(\mathbf{b}) = \mathbf{0}. $$
Now observe that
$$ \deg \left( L_0 (LF - Q) \right) \leq 2t + k - 1 < n. $$
By Lemma \ref{lemma eval connection} and Theorem \ref{th lagrange interpolation}, we conclude that
$$ L_0 (LF - Q) = 0. $$
Since $ L_0 $ is non-zero and monic, we conclude that $ LF = Q $ and we are done.
\end{proof}

Finally, once we find non-zero skew polynomials $ L, Q \in S[x;\sigma] $ such that $ L $ is monic, $ \deg(L) \leq t $, $ \deg(Q) \leq t + k - 1 $ and (\ref{eq ch 2 decoder key equation 2}) holds, then we may find $ F $ by left Euclidean division, since $ Q = LF $ by Lemma \ref{lemma ch 2 decoder lemma 2} above. Observe that left Euclidean division is possible in $ S[x;\sigma] $ since $ \sigma $ is invertible. Finding $ L $ and $ Q $ using $ R $ and $ \mathbf{b} $ (which are known) amounts to solving a system of linear equations derived from (\ref{eq ch 2 decoder key equation 2}) using the Smith normal form, as in the Gabidulin case, see \cite[Sec. III-D]{kamche}. Using this method, the decoding algorithm has an overall complexity of $ \mathcal{O}(n^3) $ operations over the ring $ S $.

\section{A Quadratic Syndrome Decoder} \label{sec syndrome decoder}

In this section, we extend the syndrome decoder from \cite{sven-efficient} to linearized Reed--Solomon codes when $ q-1 $ and $ m $ are coprime. This decoder also constitutes the first known syndrome decoder for linearized Reed--Solomon codes over finite fields, to the best of our knowledge. Note that the algorithm \cite[Alg. 2]{sven-efficient}, and the skew polynomial version (\cite[Alg. 1]{sven-efficient}) of the Byrne-Fitzpatrick algorithm \cite{byrne-fitzpatrick} it is based upon, are given in those works for Galois rings, a particular case of finite chain rings (and not all finite chain rings are Galois rings, see \cite[Th. XVII.5]{mcdonald}). However, we notice that such algorithms work for finite chain rings in general. For such a generalization, we need the following observation. For the finite chain ring $ S $, there exists an element $ \pi \in S $ such that the maximal ideal of $ S $ is $ \mathfrak{M} = (\pi) $, and all ideals of $ S $ are of the form $ \mathfrak{M}^i = (\pi^i) $, for $ i \in [r] $, where $ r $ is the smallest positive integer such that $ \pi^r = 0 $, and thus $ \mathfrak{M}^r = 0 $, see \cite[Sec. II-B]{kamche}. With this representation of the ideal chain of $ S $, one can extend mutatis mutandis \cite[Alg. 1]{sven-efficient} and the proof of its correctness and complexity to general finite chain rings. For the convenience of the reader, we include \cite[Alg. 1]{sven-efficient} for a finite chain ring $ S $ in Algorithm \ref{alg:skewbyrnefitzpatrick}. Here, we also denote $ {\rm lt}(F) = x^{\deg(F)} $, for $ F \in S[x;\sigma] $, and $ \prec $ is any total order in the set $ \{ (x^n,0) \mid n \in \mathbb{N} \} \cup \{ (0,x^n) \mid n \in \mathbb{N} \} $ compatible with multiplication by $ x^k $, for all $ k \in \mathbb{N} $. For left Gröbner bases, see \cite[Sec. III]{sven-efficient}. Finally, $ {\rm mod} $ denotes modulo on the right, that is, we say that $ F \equiv G $ mod $ H $ if, and only if, $ H $ divides $ F-G $ on the right.

\begin{algorithm}[ht!]
\caption{$\mathsf{SkewByrneFitzpatrick}$ \cite[Alg. 1]{sven-efficient}}\label{alg:skewbyrnefitzpatrick}
\SetKwInOut{Input}{Input}
\SetKwInOut{Output}{Output}
\Input{$ U \in S[x;\sigma] $ and $m \in \mathbb{Z}_{>0}$.}
\Output{Left Gröbner basis of the left $S[x;\sigma]$-module
\begin{align*}
\mathcal{M} := \left\{ (F, G) \in S[x;\sigma]^2 \mid F U \equiv G \bmod x^m \right\}.
\end{align*}}
let $\mathcal{B}_0 := \left\{ (\pi^i, 0) \mid i \in \{ 0,1, \dots, r - 1 \} \right\} \cup \left\{ (0, \pi^i) \mid i \in \{ 0,1, \dots, r - 1 \} \right\}$ \\
\For{$k \in \{ 0,1, \dots, m - 1 \}$}
{\For{each $(F_i, G_i) \in \mathcal{B}_k$}
  {compute the discrepancy $\zeta_i := (F_i U - G_i)_k$ (where $(\cdot)_k$ denotes the $k$th coefficient)}
  \For{each $(F_i, G_i) \in \mathcal{B}_k$}
  {\If{$\zeta_i = 0$}{put $(F_i, G_i) \in \mathcal{B}_{k+1}$ \\ continue}
    \eIf{there is $(F_j, G_j) \in \mathcal{B}_k$ with ${\rm lt}(F_j, G_j) \prec {\rm lt}(F_i, G_i)$ and $\zeta_j$ divides $\zeta_i$}
    {put $(F_i, G_i) - Q (F_j, G_j)$ in $\mathcal{B}_{k+1}$, where $Q \in S$ with $\zeta_i = Q \zeta_j$}
    {put $(x F_i, x G_i)$ in $\mathcal{B}_{k+1}$}}}
\Return{$\mathcal{B}_m$}
\end{algorithm} 

Throughout this section, we fix a pair $ (\mathbf{a},\boldsymbol\beta) $ as in Definition \ref{def P-independent}. We will assume that $ a_1, a_2, \ldots, a_\ell \in R^* $ satisfy $ a_i - a_j \in R^* $ for all $ 1 \leq i < j \leq \ell $ (i.e., $ \overline{a}_1, \overline{a}_2 ,\ldots, \overline{a}_\ell \in \mathbb{F}_q^* $ are all distinct) and that $ \beta_{i,1}, \beta_{i,2}, \ldots , \beta_{i,n_i} \in S $ are $ R $-linearly independent, for $ i \in [\ell] $. Hence $ (\mathbf{a},\boldsymbol\beta) $ satisfies the MSRD property by Proposition \ref{prop conjugacy repr in base field} since we are assuming that $ q-1 $ and $ m $ are coprime. In particular, fixing a dimension $ k \in [n-1] $ and a linearized Reed--Solomon code $ \mathcal{C}^\sigma_k (\mathbf{a}, \boldsymbol\beta) \subseteq S^n $, we have that $ \mathcal{C}^\sigma_k (\mathbf{a}, \boldsymbol\beta)^\perp = \mathcal{C}^{\sigma^{-1}}_{n-k} \left(\mathbf{a}, \boldsymbol\delta \right) $ by Theorem \ref{th dual of LRS}, where $ (\mathbf{a},\boldsymbol\delta) $ also satisfies the MSRD property.

We consider the same error-correcting scenario as in Section \ref{sec WB decoder}. That is, $ t = \left\lfloor (n-k)/2 \right\rfloor $ as in (\ref{eq ch 2 def t for welch-berlekamp}) and $ \mathbf{r} = \mathbf{c} + \mathbf{e} \in S^n $ as in (\ref{eq ch 2 def received word}), for a fixed codeword $ \mathbf{c} \in \mathcal{C}^\sigma_k (\mathbf{a}, \boldsymbol\beta) $ and an error vector $ \mathbf{e} \in S^n $, where we may assume that $ t = {\rm wt}_{SR}(\mathbf{e}) $. 

We start by extending \cite[Def. 2]{sven-efficient}. Notice that $ a_i^{-1} - (a_j^{-1})^\beta \in S^* $, for all $ \beta \in S^* $, with respect to $ \sigma^{-1} $, since the same property holds for $ a_i $ and $ a_j $ with respect to $ \sigma $ by assumption, for all $ 1 \leq i < j \leq \ell $.

\begin{definition}
We say that $ F = \sum_{i=0}^d F_i x^i \in S[x;\sigma^{-1}] $, $ d \in \mathbb{N} $, is primitive if it is not a zero divisor, i.e., $ F_i \in S^* $ for some $ i \in [d] $, i.e., $ \overline{F} \in \mathbb{F}_{q^m}[x;\overline{\sigma}^{-1}] \setminus \{ 0 \} $. We say that $ \Lambda \in S[x;\sigma^{-1}] $ is an annihilator of $ \mathbf{e} \in S^n $ if it is primitive, $ \Lambda_{a_i^{-1}}(e_{i,j}) = 0 $ for $ j \in [n_i] $ and $ i \in [\ell] $ and it has minimum possible degree among primitive skew polynomials in $ S[x;\sigma^{-1}] $ satisfying such a property. 
\end{definition}

Notice that, here, $ \Lambda_{a_i^{-1}}(e_{i,j}) $ is the operator evaluation (Definition \ref{def linearized operators}) with respect to $ \sigma^{-1} $. We will not specify this in the notation $ \Lambda_{a_i^{-1}}(e_{i,j}) $ since we wrote that $ \Lambda \in S[x;\sigma^{-1}] $, hence emphasizing the use of $ \sigma^{-1} $ for $ \Lambda $ instead of $ \sigma $.

We need some preliminary auxiliary properties on the zeros of skew polynomials over finite chain rings. This result extends Lemma \ref{lemma zeros lam leroy} in a different direction than Theorem \ref{th zeros rings}.

\begin{lemma} \label{lemma degree bound with rank}
If $ F \in S[x;\sigma] $ is primitive and $ c_1, c_2, \ldots, c_\ell \in S^* $ are such that $ c_i - c_j^\beta \in S^* $ for all $ \beta \in S^* $ and all $ 1 \leq i < j \leq \ell $, then
$$ \sum_{i=1}^\ell {\rm rk}_R(\ker(F_{c_i})) \leq \deg(F). $$
\end{lemma}
\begin{proof}
Let $ r_i = \ker(F_{c_i}) $, for $ i \in [\ell] $. Using the Smith normal form, we see that there are $ R $-linearly independent elements $ b_{i,1}, \ldots, b_{i,r_i} \in S $ and non-zero $ \lambda_{i,1}, \ldots, \lambda_{i,r_i} \in R $ such that $ \ker(F_{c_i}) = \langle \lambda_{i,1}b_{i,1}, \ldots, \lambda_{i,r_i}b_{i,r_i} \rangle_R $, for $ i \in [\ell] $. Since $ R $ is a chain ring, we may assume that there exists $ k \in [\ell] $ such that $ \lambda_{i,j} | \lambda_{k,r_k} $ for all $ j \in [r_i] $ and all $ i \in [\ell] $. Therefore, since $ F_{c_i}(\lambda_{i,j} b_{i,j}) = \lambda_{i,j} F_{c_i}(b_{i,j}) = 0 $, we see that $ (\lambda_{k,r_k} F)_{c_i}(b_{i,j}) = 0 $, for all $ j \in [r_i] $ and all $ i \in [\ell] $. 

Assume that $ \deg(F) < r_1 + r_2 + \cdots + r_\ell $. Then we deduce that $ \lambda_{k,r_k} F = 0 $ by Theorem \ref{th zeros rings}. However, since $ \lambda_{k,r_k} \neq 0 $, then $ F $ is not primitive, a contradiction. Therefore, $ \deg(F) \geq r_1 + r_2 + \cdots + r_\ell $ and we are done.
\end{proof}

We next extend \cite[Lemma 4]{sven-efficient}.

\begin{lemma}
Any annihilator of $ \mathbf{e} \in S^n $ has degree $ t = {\rm wt}_{SR}(\mathbf{e}) $. In addition, if $ {\rm rk}(\mathbf{e}_i) = {\rm frk}(\mathbf{e}_i) $, for $ i \in [\ell] $, then there is a unique monic annihilator of $ \mathbf{e} $.
\end{lemma}
\begin{proof}
Let $ \Lambda \in S[x;\sigma^{-1}] $ be an annihilator of $ \mathbf{e} $. First, $ \deg(\Lambda) \leq t $ by Corollary \ref{cor annihilator skew pol from rank}. Second, if $ \deg(\Lambda) < t $, then $ \Lambda $ would not be primitive by Lemma \ref{lemma degree bound with rank}. Hence $ \deg(\Lambda) = t $. 

Now assume that $ {\rm rk}(\mathbf{e}_i) = {\rm frk}(\mathbf{e}_i) $, for $ i \in [\ell] $. First, there exists a monic annihilator $ \Lambda \in S[x;\sigma^{-1}] $ of $ \mathbf{e} $ by Corollary \ref{cor annihilator skew pol from rank}. Let $ \Lambda^\prime \in S[x;\sigma^{-1}] $ be another annihilator of $ \mathbf{e} $. Note that $ t = \deg(\Lambda) = \deg(\Lambda^\prime) $. Since $ \Lambda $ is monic, we may perform right Euclidean division, i.e., there are $ Q,R \in S[x;\sigma^{-1}] $ with $ \deg(R) < t $ and $ \Lambda^\prime = Q \Lambda + R $. By Lemmas \ref{lemma product rule} and \ref{lemma eval connection}, we have that $ R_{a_i^{-1}}(e_{i,j}) = 0 $, for $ j \in [n_i] $ and $ i \in [\ell] $. Since $ \deg(R) < \sum_{i=1}^\ell {\rm frk}_R(\mathbf{e}_i) $, we deduce that $ R = 0 $ by Theorem \ref{th lagrange interpolation}. In other words, $ \Lambda^\prime = Q \Lambda $, where $ Q \in S^* $, and thus $ \Lambda $ is the unique monic annihilator of $ \mathbf{e} $.
\end{proof}

We will define syndromes as usual.

\begin{definition}
Let $ h = n-k $ and define the syndrome vector $ \mathbf{s} = \mathbf{e} \mathbf{M}^{\sigma^{-1}}_h(\mathbf{a},\boldsymbol\delta)^\intercal \in S^h $. We define the syndrome skew polynomial $ s = \sum_{i=0}^{h-1} s_i x^i \in S[x;\sigma^{-1}] $, where $ \mathbf{s} = (s_0,s_1, \ldots, s_{h-1}) $.
\end{definition}

In order to prove the key equation between annihilators of $ \mathbf{e} $ and the syndrome skew polynomial $ s $, we need the following two lemmas. The first one follows directly from the Smith normal form.

\begin{lemma} \label{lemma decomposition error}
For $ i \in [\ell] $, there exist $ \boldsymbol\alpha_i \in S^{t_i} $ and $ \mathbf{B}^{(i)} \in R^{t_i \times n_i} $ such that $ \mathbf{e}_i = \boldsymbol\alpha_i \mathbf{B}^{(i)} $, $ t_i = {\rm frk}(\boldsymbol\alpha_i) = {\rm rk}(\mathbf{B}^{(i)}) $ and $ \langle \alpha_{i,1}, \ldots, \alpha_{i,t_i} \rangle_R = \langle e_{i,1}, \ldots, e_{i,n_i} \rangle_R $. In particular, we have that $ \Lambda_{a_i^{-1}}(\alpha_{i,j}) = 0 $, for all $ j \in [t_i] $ and $ i \in [\ell] $, for any annihilator $ \Lambda \in S[x;\sigma] $ of $ \mathbf{e} $.
\end{lemma}

The second lemma can be found in \cite[Lemma 1]{lam}.
 
\begin{lemma} [\textbf{\cite{lam}}] \label{lemma lam 1}
For all $ a \in S $ and all integers $ 0 \leq j \leq i $, it holds that
$$ N^\sigma_i(a) = \sigma^{i-j} \left( N^\sigma_j(a) \right) N^\sigma_{i-j}(a). $$
\end{lemma}

We may now provide the key equation.

\begin{theorem} [\textbf{Key Equation}] \label{th key equation}
Let $ \Lambda \in S[x;\sigma^{-1}] $ be an annihilator of $ \mathbf{e} $. There exists $ \Omega \in S[x;\sigma^{-1}] $ with $ \deg(\Omega) <t $ and 
\begin{equation}
\Omega \equiv \Lambda s \quad {\rm mod} \quad x^h.
\label{eq key equation}
\end{equation}
\end{theorem}
\begin{proof}
If we set $ d_{i,u} = \sum_{j=1}^{n_i} \mathbf{B}^{(i)}_{u,j} \delta_{i,j} \in S $, for $ u \in [t_i] $ and $ i \in [\ell] $, then we have that
\begin{equation}
\begin{split}
s_v & = \sum_{i=1}^\ell \sum_{j=1}^{n_i} e_{i,j} \mathcal{D}_{\sigma^{-1},a_i}^v \left( \delta_{i,j} \right) \\
 & = \sum_{i=1}^\ell \sum_{j=1}^{n_i} \sum_{u=1}^{t_i} \alpha_{i,u} \mathbf{B}^{(i)}_{u,j} \mathcal{D}_{\sigma^{-1},a_i}^v \left( \delta_{i,j} \right) \\
 & = \sum_{i=1}^\ell \sum_{u=1}^{t_i} \sum_{j=1}^{n_i} \alpha_{i,u} \mathcal{D}_{\sigma^{-1},a_i}^v \left( \mathbf{B}^{(i)}_{u,j} \delta_{i,j} \right) \\
 & = \sum_{i=1}^\ell \sum_{u=1}^{t_i} \alpha_{i,u} \mathcal{D}_{\sigma^{-1},a_i}^v \left( d_{i,u} \right) ,
\end{split}
\label{eq formula for syndrome}
\end{equation}
for $ v = 0,1, \ldots, h-1 $.

Note that $ \deg (\Lambda) = t $ and $ \deg(s) = h-1 $, and therefore, $ \Lambda s = \sum_{v=0}^{t+h-1} (\Lambda s)_v $. For $ v = t, t+1, \ldots, h-1 $, we have that
\begin{equation*}
\resizebox{\columnwidth}{!}{%
$
\begin{split}
(\Lambda s)_v & = \sum_{l=0}^v \Lambda_{v-l} \sigma^{-v+l} (s_l) \\
& \stackrel{(a)}{=} \sum_{i=1}^\ell \sum_{u=1}^{t_i} \sum_{l=0}^v \Lambda_{v-l} \sigma^{-v+l} \left( \alpha_{i,u} \mathcal{D}_{\sigma^{-1},a_i}^l \left( d_{i,u} \right) \right) \\
& = \sum_{i=1}^\ell \sum_{u=1}^{t_i} \sum_{l=0}^v \Lambda_{v-l} \sigma^{-v+l} \left( \alpha_{i,u} \right) \sigma^{-v+l} \left( N^{\sigma^{-1}}_l (a_i) \right) \sigma^{-v+l} \left( \sigma^{-l} \left( d_{i,u} \right) \right) \\
& \stackrel{(b)}{=} \sum_{i=1}^\ell \sum_{u=1}^{t_i} \sum_{l=0}^v \Lambda_{v-l} \sigma^{-v+l} \left( \alpha_{i,u} \right) N^{\sigma^{-1}}_{v-l} \left( a_i^{-1} \right) N^{\sigma^{-1}}_v (a_i) \sigma^{-v} \left( d_{i,u} \right) \\
& = \sum_{i=1}^\ell \sum_{u=1}^{t_i} \sum_{l=0}^v \Lambda_{v-l} \mathcal{D}_{\sigma^{-1},a_i^{-1}}^{v-l} \left( \alpha_{i,u} \right) \mathcal{D}_{\sigma^{-1},a_i}^v \left( d_{i,u} \right) \\
& = \sum_{i=1}^\ell \sum_{u=1}^{t_i} \mathcal{D}_{\sigma^{-1},a_i}^v \left( d_{i,u} \right) \left( \sum_{l=0}^v \Lambda_l \mathcal{D}_{\sigma^{-1},a_i^{-1}}^l \left( \alpha_{i,u} \right) \right) \\
& \stackrel{(c)}{=} \sum_{i=1}^\ell \sum_{u=1}^{t_i} \mathcal{D}_{\sigma^{-1},a_i}^v \left( d_{i,u} \right) \left( \sum_{l=0}^t \Lambda_l \mathcal{D}_{\sigma^{-1},a_i^{-1}}^l \left( \alpha_{i,u} \right) \right) \\
& \stackrel{(d)}{=} 0,
\end{split}
$
}%
\end{equation*}
where we have used the formula (\ref{eq formula for syndrome}) for $ s_l $ in (a) since $ l \leq v \leq h-1 $, Lemma \ref{lemma lam 1} in (b), the fact that $ \Lambda_l = 0 $ if $ t < l \leq v $ in (c), and Lemma \ref{lemma decomposition error} in (d).
\end{proof}

The next ingredient is the following extension of \cite[Th. 7]{sven-decoder}. From now on, we will also define
$$ \widetilde{\beta}_{i,j} = \sigma^{k-1}(\beta_{i,j}) a_i^{k-1}, $$
for $ j \in [n_i] $ and $ i \in [\ell] $. Note that $ \widetilde{\beta}_{i,1}, \ldots , \widetilde{\beta}_{i,n_i} \in S^* $ are also $ R $-linearly independent, since $ a_i \in R^* $, $ \sigma $ is an automorphism and $ \beta_{i,1}, \ldots , \beta_{i,n_i} \in S^* $ are $ R $-linearly independent.

\begin{theorem} \label{th pair U V}
Recall that $ \mathbf{r} = \mathbf{c} + \mathbf{e} $, where $ \mathbf{c} \in \mathcal{C}^\sigma_k (\mathbf{a}, \boldsymbol\beta) $ and $ \mathbf{e} \in S^n $ is such that $ t = {\rm wt}_{SR}(\mathbf{e}) $. Assume that we have non-zero $ U,V \in S[x;\sigma^{-1}] $ such that
\begin{enumerate}
\item
$ U $ is primitive,
\item
$ Us - V \equiv 0 $ $ {\rm mod} $ $ x^h $,
\item
$ \deg(U) \leq t $,
\item
$ \deg(V) < \deg(U) $.
\end{enumerate}
Then $ U $ is an annihilator of $ \mathbf{e} $ and in particular, $ \deg(U) = t $. Moreover,
$$ UR \equiv U \widetilde{F} \quad {\rm mod} \quad G, $$
where $ R, \widetilde{F},G \in S[x;\sigma^{-1}] $ are the unique skew polynomials with
$$
\begin{array}{rlcl}
R_{a_i^{-1}}(\widetilde{\beta}_{i,j}) & = r_{i,j} \quad & \textrm{and} \quad & \deg(R) < n, \\
\widetilde{F}_{a_i^{-1}}(\widetilde{\beta}_{i,j}) & = c_{i,j} \quad & \textrm{and} \quad & \deg(\widetilde{F}) < k, \\
G_{a_i^{-1}}(\widetilde{\beta}_{i,j}) & = 0 \quad & \textrm{and} \quad & \deg(G) = n ,
\end{array}
$$
and $ G $ is the unique monic annihilator of $ \widetilde{\boldsymbol\beta} \in S^n $.
\end{theorem}
\begin{proof}
Since $ \deg(V) < t $, $ 2t - 1 < h $ and $ Us - V \equiv 0 $ mod $ x^h $, then $ (Us)_i = 0 $, for $ i = t,t+1, \ldots , 2t-1 $. This may be rewritten as
$$ \resizebox{\columnwidth}{!}{%
$ \left( \begin{array}{cccc}
\sigma^0(s_t) & \sigma^{-1}(s_{t-1}) & \ldots & \sigma^{-t}(s_0) \\
\sigma^0(s_{t+1}) & \sigma^{-1}(s_t) & \ldots & \sigma^{-t}(s_1) \\
\vdots & \vdots & \ddots & \vdots \\
\sigma^0(s_{2t-1}) & \sigma^{-1}(s_{2t-2}) & \ldots & \sigma^{-t}(s_{t-1}) 
\end{array} \right) \left( \begin{array}{c}
u_0 \\
u_1 \\
\vdots \\
u_t
\end{array} \right) = \left( \begin{array}{c}
0 \\
0 \\
\vdots \\
0
\end{array} \right),
$
}%
$$
where $ U = \sum_{i=0}^t u_ix^i $. If we denote by $ \mathbf{S} \in S^{t \times (t+1)} $ the matrix above, then by (\ref{eq formula for syndrome}) we have a decomposition
$$ \mathbf{S} = \mathbf{D} \mathbf{A}, $$
where 
$$ 
\resizebox{\columnwidth}{!}{%
$
\mathbf{D} = (\mathbf{D}_1 | \ldots | \mathbf{D}_\ell) \in S^{t \times t}, \quad \mathbf{A} = \left( \begin{array}{c}
\mathbf{A}_1 \\
\hline
\vdots \\
\hline
\mathbf{A}_\ell
\end{array} \right) \in S^{t \times (t+1)},
$
}%
$$
and
$$ 
\resizebox{\columnwidth}{!}{%
$
\mathbf{D}_i = \left( \begin{array}{cccc}
\mathcal{D}_{\sigma^{-1},a_i}^t (d_{i,1}) & \mathcal{D}_{\sigma^{-1},a_i}^t (d_{i,2}) & \ldots & \mathcal{D}_{\sigma^{-1},a_i}^t (d_{i,t_i}) \\
\mathcal{D}_{\sigma^{-1},a_i}^{t+1} (d_{i,1}) & \mathcal{D}_{\sigma^{-1},a_i}^{t+1} (d_{i,2}) & \ldots & \mathcal{D}_{\sigma^{-1},a_i}^{t+1} (d_{i,t_i}) \\
\vdots & \vdots & \ddots & \vdots \\
\mathcal{D}_{\sigma^{-1},a_i}^{2t-1} (d_{i,1}) & \mathcal{D}_{\sigma^{-1},a_i}^{2t-1} (d_{i,2}) & \ldots & \mathcal{D}_{\sigma^{-1},a_i}^{2t-1} (d_{i,t_i}) \\
\end{array} \right) \in S^{t \times t_i},
$
}%
$$
$$
\resizebox{\columnwidth}{!}{%
$
\mathbf{A}_i = \left( \begin{array}{cccc}
\mathcal{D}_{\sigma^{-1},a_i^{-1}}^0 (\alpha_{i,1}) & \mathcal{D}_{\sigma^{-1},a_i^{-1}}^1 (\alpha_{i,1}) & \ldots & \mathcal{D}_{\sigma^{-1},a_i^{-1}}^t (\alpha_{i,1}) \\
\mathcal{D}_{\sigma^{-1},a_i^{-1}}^0 (\alpha_{i,2}) & \mathcal{D}_{\sigma^{-1},a_i^{-1}}^1 (\alpha_{i,2}) & \ldots & \mathcal{D}_{\sigma^{-1},a_i^{-1}}^t (\alpha_{i,2}) \\
\vdots & \vdots & \ddots & \vdots \\
\mathcal{D}_{\sigma^{-1},a_i^{-1}}^0 (\alpha_{i,t_i}) & \mathcal{D}_{\sigma^{-1},a_i^{-1}}^1 (\alpha_{i,t_i}) & \ldots & \mathcal{D}_{\sigma^{-1},a_i^{-1}}^t (\alpha_{i,t_i}) \\
\end{array} \right) \in S^{t_i \times (t+1)},
$
}%
$$
for $ i \in [\ell] $. Since $ d_{i,1} ,\ldots d_{i,t_i} \in S $ are $ R $-linearly independent, we deduce from Theorem \ref{th ext moore matrix invertible} that $ \mathbf{D} \in S^{t \times t} $ is invertible. Thus we have that 
$$ \mathbf{S} \mathbf{u} = \mathbf{D} \mathbf{A} \mathbf{u} = \mathbf{0} \quad \Longleftrightarrow \quad \mathbf{A} \mathbf{u} = \mathbf{0}, $$
which means that $ U_{a_i^{-1}} (\alpha_{i,j}) = 0 $, for all $ j \in [t_i] $ and $ i \in [\ell] $. Since $ \langle \alpha_{i,1}, \ldots, \alpha_{i,t_i} \rangle_R = \langle e_{i,1}, \ldots, e_{i,n_i} \rangle_R $, then $ U_{a_i^{-1}}(e_{i,j}) = 0 $, for $ j \in [n_i] $ and $ i \in [\ell] $. Since $ \deg(U) \leq t $ and it is primitive, then $ U $ is an annihilator of $ \mathbf{e} $. Finally, we have
$$ 0 = U_{a_i^{-1}}(e_{i,j}) = U_{a_i^{-1}}(r_{i,j} - c_{i,j}) = $$
$$ U_{a_i^{-1}}(R_{a_i^{-1}}(\widetilde{\beta}_{i,j}) - \widetilde{F}_{a_i^{-1}}(\widetilde{\beta}_{i,j})) = (U(R-\widetilde{F}))_{a_i^{-1}}(\widetilde{\beta}_{i,j}), $$
for $ j \in [n_i] $ and $ i \in [\ell] $, using Lemmas \ref{lemma product rule} and \ref{lemma eval connection} in the last equality. Since $ G $ is a monic annihilator of $ \widetilde{\boldsymbol\beta} $, then we deduce that $ G $ divides $ U(R-\widetilde{F}) $ on the right, and we are done.
\end{proof}

Finally, we show that we may recover the skew polynomial associated to $ \mathbf{c} $ for the pair $ (\mathbf{a}, \boldsymbol\beta) $ and $ \sigma $ from that for the pair $ (\mathbf{a}^{-1},\widetilde{\boldsymbol\beta}) $ and $ \sigma^{-1} $.

\begin{proposition} \label{prop from F line to F}
Let $ F = \sum_{u=0}^{k-1} F_u x^u \in S[x;\sigma] $ and $ \widetilde{F} = \sum_{u=0}^{k-1} \widetilde{F}_u x^u \in S[x;\sigma^{-1}] $ be related by $ \widetilde{F}_{k-u-1} = F_u $, for $ u = 0 ,1,\ldots, k-1 $. Then
$$ F_{a_i}(\beta_{i,j}) = \widetilde{F}_{a_i^{-1}}(\widetilde{\beta}_{i,j}), $$
where $ \widetilde{\beta}_{i,j} = \sigma^{k-1}(\beta_{i,j}) a_i^{k-1} $, for $ j \in [n_i] $ and $ i \in [\ell] $. 
\end{proposition}
\begin{proof}
Since $ a_i \in R^* $ and $ \widetilde{F}_{a_i^{-1}} $ is $ R $-linear, then for $ j \in [n_i] $ and $ i \in [\ell] $, we have
\begin{equation*}
\begin{split}
\widetilde{F}_{a_i^{-1}}(\widetilde{\beta}_{i,j}) & = \widetilde{F}_{a_i^{-1}}(a^{k-1}\sigma^{k-1}(\beta_{i,j})) \\
& = a^{k-1} \widetilde{F}_{a_i^{-1}}(\sigma^{k-1}(\beta_{i,j})) \\ 
& = a^{k-1} \sum_{u=0}^{k-1} \widetilde{F}_u \sigma^{-u}(\sigma^{k-1}(\beta_{i,j}))a_i^{-u} \\
& = \sum_{u=0}^{k-1} F_{k-u-1} \sigma^{k-u-1}(\beta_{i,j})a_i^{k-u-1} \\
& = F_{a_i}(\beta_{i,j}).
\end{split}
\end{equation*}
\end{proof}

The algorithm in \cite[alg. 2]{sven-decoder} can be extended to our case as shown in Algorithm \ref{alg:decoder}. In the following theorem, we prove its correctness and give its complexity.

\begin{algorithm}[ht!]
\caption{$\mathsf{SyndromeDecoder}$}\label{alg:decoder}
\SetKwInOut{Input}{Input}
\SetKwInOut{Output}{Output}
\Input{$\mathbf{r} \in S^n$} %
\Output{%
If there is a $\mathbf{c} = \left( F_{a_1}(\beta_{1,1}),\dots, F_{a_\ell}(\beta_{\ell, n_\ell})\right) \in \mathcal{C}^\sigma_k (\mathbf{a}, \boldsymbol\beta) $ with $ F \in S[x;\sigma]$, $ \deg(F) < k $ and ${\rm d}_{SR}(\mathbf{r},\mathbf{c}) \leq \tfrac{n-k}{2}$, then $F$.\\ Otherwise ``decoding failure''.}

$ \mathbf{s} := \mathbf{H} \mathbf{r}^\intercal $ \\
$s := \sum_{i=0}^{n-k-1} s_i x^i \in S[x;\sigma^{-1}] $ \\
$\mathcal{B} := \mathsf{SkewByrneFitzpatrick}(s, n-k)$ \\
$(\Lambda,\Omega) :=$ element of $\mathcal{B}$ of minimal degree among all $(U,V) \in \mathcal{B}$ with $\deg (U) > \deg (V) $ and $U$ primitive. \\
$R :=$ unique $ R \in S[x;\sigma^{-1}] $ such that $ R_{a_i^{-1}}(\widetilde{\beta}_{i,j}) = r_{i,j} $, for all $ i,j $, with $ \deg(R) < n $. \label{line:R} \\
$G :=$ unique $ G \in S[x;\sigma^{-1}] $ such that $ G_{a_i^{-1}}(\widetilde{\beta}_{i,j}) = 0 $, for all $ i,j $, with $ \deg(G) < n $. \label{line:G} \\
$\Psi := \Lambda R \textrm{ mod} \textrm{ } G$ \\
$(\widetilde{F},T) :=$ quotient and remainder of left division of $\Psi$ by $\Lambda$. \\
$ F := $ skew polynomial obtained by $ \widetilde{F}_{k-u-1} = F_u $, for $ u = 0 ,1,\ldots, k-1 $. \\
\If{$T =0$ and ${\rm d}_{SR} \left( \mathbf{r}, \left( F_{a_1}(\beta_{1,1}),\dots, F_{a_\ell}(\beta_{\ell, n_\ell})\right) \right) \leq \tfrac{n-k}{2}$ and $\deg (F) < k$\label{line:check_if_f_has_correct_form}}{
\Return{$F$}
} 
\Else 
{
\Return{``decoding failure''}
}
\end{algorithm}

\begin{theorem}
Algorithm \ref{alg:decoder} is correct and has a complexity of $ \mathcal{O}(r n^2) $ operations in $ S $, where $ r $ is the smallest positive integer such that $ \pi^r = 0 $ or $ \mathfrak{M}^r = 0 $.
\end{theorem}
\begin{proof}
Using Algorithm \ref{alg:skewbyrnefitzpatrick}, we obtain a pair $ U,V \in S[x;\sigma^{-1}] $ with $ U $ primitive, $ \deg(U) > \deg(V) $, $ U s \equiv V $ mod $ x^h $ and $ \deg(U) $ minimal among pairs with these properties. By Theorem \ref{th key equation}, there is a pair $ (\Lambda, \Omega) $ satisfying such properties and with $ \deg(\Lambda) = t $. Thus we have $ \deg(U) \leq t $.

Since $ t \leq h/2 $, then $ U $ is an annihilator of $ \mathbf{e} $ by Theorem \ref{th pair U V}. Moreover, we have $ UR \equiv U \widetilde{F} \textrm{ mod} $ $ G $, with notation as in Theorem \ref{th pair U V}. 

Since $ \deg( U \widetilde{F}) = \deg(U) + \deg(\widetilde{F}) < t+k-1 < n = \deg(G) $, then we may obtain $ U \widetilde{F} $ by right division of $ UR $ by $ G $. Note that $ R $ and $ G $ may be computed from the received word $ \mathbf{r} $ and the pair $ (\mathbf{a}, \boldsymbol\beta) $, and that $ G $ is monic, hence right division by $ G $ is well defined. Next, since $ U $ is primitive, then we may divide $ U \widetilde{F} \neq 0 $ by $ U $ on the left and we obtain $ \widetilde{F} $. Finally, we compute $ F $ from $ \widetilde{F} $ as in Proposition \ref{prop from F line to F}, where $ F $ is the skew polynomial whose coefficients contain the message encoded by the sent codeword $ \mathbf{c} $, and we are done.
 
Finally, the complexity of the skew polynomial Byrne-Fitzpatrick Algorithm \ref{alg:skewbyrnefitzpatrick} has a complexity of $ \mathcal{O}(r n^2) $ operations in $ S $ by \cite[Th. 3]{sven-decoder} (the extension from Galois rings to general finite chain rings is straightforward as mentioned at the beginning of the section). The other operations that appear in Algorithm \ref{alg:decoder} can be implemented with a complexity of $ \mathcal{O}(n^2) $ operations in $ S $ by \cite[Lemma 8]{sven-decoder}. 
\end{proof}

\section{Applications} \label{sec applications}

In this section, we briefly discuss applications of MSRD codes over finite chain rings, and in particular, the linearized Reed--Solomon codes from Definition \ref{def LRS codes}. We will only focus on applications in Space-Time Coding and Multishot Network Coding, and we will only briefly discuss how to adapt ideas from the literature to the case of MSRD codes over finite chain rings.

\subsection{Space-Time Coding with Multiple Fading Blocks} \label{subsec space time}

Space-time codes \cite{tarokh} are used in wireless communication, in scenarios of multiple input/multiple output antenna transmission. Such codes utilize space diversity (via multiple antennas) and time diversity (via interleaving up to some delay constraint) in order to reduce the fading of the channel. 

In the case of one fading block, codewords are seen as matrices in $ \mathcal{A}^{n_t \times T} $, where $ \mathcal{A} \subseteq \mathbb{C} $ is the signal constellation (a subset of the complex field), $ n_t $ is the number of transmit antennas and $ T $ is the time delay. In particular, the code is a subset $ \mathcal{C} \subseteq \mathcal{A}^{n_t \times T} $. In this scenario, the code achieves transmit diversity gain $ d $ (or simply code diversity) if the rank of the difference of any two matrices in the code is at least $ d $, see \cite{tarokh, space-time-kumar}. Large code diversity is desirable, but it competes with the symbol rate of the code, defined as 
$$ \frac{1}{T} \log_{|\mathcal{A}|} |\mathcal{C}|. $$
The symbol rate is an important parameter when the constellation $ \mathcal{A} $ is constrained or we wish it to be as small as possible. See the discussion in \cite{Mohannad-Thesis}. The diversity-rate tradeoff is expressed in a Singleton-type bound, and codes attaining equality in such a bound may be obtained by mapping a maximum rank distance (MRD) code over a finite field, such as a Gabidulin code, into the constellation $ \mathcal{A} \subseteq \mathbb{C} $. This may be done via Gaussian integers \cite{lusina} or Eisenstein integers \cite{sven-ST}. 

The case of multiple fading blocks, say $ L $, was first investigated in \cite{space-time-kumar}. In this case, the codewords are matrices in $ \mathcal{A}^{n_t \times LT} $, which can be thought of as $ L $ matrices of size $ n_t \times T $, that is tuples in $ (\mathcal{A}^{n_t \times T})^L $. In this case, a code diversity $ d $ is attained if the minimum sum-rank distance of the code is at least $ d $. For this reason, space-time codes with optimal rate-diversity tradeoff in the multiblock case may be obtained by mapping MSRD codes over finite fields to the constellation $ \mathcal{A} \subseteq \mathbb{C} $. This was observed in \cite{space-time-kumar}, and linearized Reed--Solomon codes were first used for this purpose in \cite{mohannad}. As shown there, the use of linearized Reed--Solomon codes allows one to attain optimal rate-diversity while minimizing the time delay $ T $, and while the constellation size $ |\mathcal{A}| $ grows linearly in $ L $, in contrast with previous space-time codes, whose constellation sizes grow exponentially in $ L $. See also \cite{Mohannad-Thesis}.

In \cite[Sec. VI-A]{kamche}, it was shown how to translate any MRD code over a finite principal ideal ring into a space-time code over a complex constellation with optimal rate-diversity tradeoff. This result can be immediately generalized to the multiblock case by using MSRD codes over finite principal ideal rings, as follows. We omit the proof.

\begin{theorem}
Let $ Q $ be a principal ideal ring such that there exists a surjective ring morphism $ \varphi : Q \longrightarrow R $ (recall that $ R $ is a finite chain ring). Let $ \varphi^* : R \longrightarrow Q $ be such that $ \varphi \circ \varphi^* = {\rm Id} $. Extend both maps component-wise to tuples of matrices in $ (Q^{n_t \times T})^L $ and $ (R^{n_t \times T})^L $. If $ \mathcal{C} \subseteq (R^{n_t \times T})^L $ is an MSRD code, then so is $ \varphi^*(\mathcal{C}) \subseteq (Q^{n_t \times T})^L $, of the same dimension and minimum sum-rank distance. In particular, if $ Q \subseteq \mathbb{C} $, then $ \varphi^*(\mathcal{C}) $ is a space-time code with optimal rate-diversity tradeoff for $ L $ fading blocks.
\end{theorem}

Examples may be constructed easily. For instance, we may choose $ Q = \mathbb{Z}[i] $ and $ R = \mathbb{Z}_{2^r}[i] = \mathbb{Z}[i] / (2^r) $, where $ r $ is a positive integer (see also the introduction). In this case, we have that $ q = 2^r $, and we may construct linearized Reed--Solomon codes corresponding to $ L = 2^r - 1 $ fading blocks.

\subsection{Physical-Layer Multishot Network Coding} \label{subsec multishot}

Linear network coding \cite{linearnetwork} permits maximum information flow from a source to several sinks (\textit{multicast}) in one use of the network (a \textit{single shot}). In such a communication scenario, MRD codes can correct a given number of link errors and packet erasures with the maximum possible information rate, without knowledge and independently of the transfer matrix or topology of the network (\textit{universal error correction}), see \cite{errors-network, error-control}. In the case of a number $ \ell $ of uses of the network (\textit{multishot Network Coding}), the minimum sum-rank distance of the code determines how many link errors and packet erasures the code can correct in total throughout the $ \ell $ shots of the network, without knowledge and independently of the transfer matrices and network topology \cite{multishot}. MSRD codes, in particular linearized Reed--Solomon codes, over finite fields were used in this scenario in \cite{secure-multishot}. 

In \cite{feng1}, a similar model was developed for physical-layer Network Coding, where the network code lies in some constellation (a subset of the complex numbers), which may be identified with finite principal ideal rings of the form $ \mathbb{Z}[i] / (q) $, for a positive integer $ q $. Such rings are finite chain rings if $ q = 2^r $ and $ r $ is a positive integer, see the introduction.

Just as in the case of finite fields, we may consider $ \ell $ shots in physical-layer Network Coding, with network-code alphabet $ R = \mathbb{Z}[i] / (q) $, which is a finite chain ring if $ q = 2^r $ as above. In this case, codes are subsets $ \mathcal{C} \subseteq R^{m \times n} $, where $ m $ is the packet length and $ n $ is the number of outgoing links from the source. Using the matrix representation map (\ref{eq def matrix representation map}), we may consider $ \mathcal{C} \subseteq S^n $, and we may then choose $ \mathcal{C} $ to be $ S $-linear, such as a linearized Reed--Solomon code (Definition \ref{def LRS codes}), which may attain the value $ \ell = q-1 = 2^r-1 $. If $ \mathbf{c} \in \mathcal{C} $ is transmitted, then the output of an $ \ell $-shot linearly coded network over $ R $ with at most $ t $ link errors and $ \rho $ packet erasures is
$$ \mathbf{y} = \mathbf{c} \mathbf{A} + \mathbf{e} \in S^N, $$
where $ \mathbf{e} \in S^N $ is such that $ \wt_{SR}(\mathbf{e}) \leq t $, and $ \mathbf{A} = \diag(\mathbf{A}_1, \mathbf{A}_2, \ldots, \mathbf{A}_\ell) $ is such that $ \frk(\mathbf{A}) \geq n - \rho $, where $ \mathbf{A}_i \in R^{n_i \times N_i} $ is the transfer matrix of the $ i $th shot, for $ i \in [\ell] $, where $ N = N_1 + N_2 + \cdots N_\ell $ and $ n = n_1 + n_2 + \cdots + n_\ell $. For simplicity, we will only consider the \textit{coherent} scenario, meaning that we assume that the transfer matrix $ \mathbf{A} $ is known to the receiver.

We may now prove that the minimum sum-rank distance of the code $ \mathcal{C} $ gives a sufficient condition for error and erasure correction in the scenario described above. We omit necessary conditions for brevity.

\begin{theorem}
In the scenario described above, fix $ \rho = n - \frk(\mathbf{A}) $. If
$$ 2t + \rho + 1 \leq \dd_{SR}(\mathcal{C}), $$
then there exists a decoder $ D_\mathbf{A} : S^n \longrightarrow \mathcal{C} $ sending $ D_\mathbf{A}(\mathbf{c} \mathbf{A} + \mathbf{e}) = \mathbf{c} $, for all $ \mathbf{c} \in \mathcal{C} $ and all $ \mathbf{e} \in S^n $ with $ \wt_{SR} (\mathbf{e}) \leq t $. In particular, if $ \mathcal{C} $ is an MSRD code over $ R $ for the length partition $ n = n_1 + n_2 + \cdots + n_\ell $, then it may correct $ t $ link errors, $ \rho $ packet erasures and achieves an information rate of 
$$ \frac{n - 2t - \rho}{n}. $$
\end{theorem}
\begin{proof}
Let $ \mathbf{C} \in R^{m \times n} $, $ \mathbf{B} \in R^{n \times N} $ and $ r = \frk(\mathbf{B}) $. We may assume that the first $ r $ columns of $ \mathbf{B} $, which form a matrix $ \mathbf{B}^\prime \in R^{n \times r} $, are $ R $-linearly independent. By \cite[p. 92, ex. V.14]{mcdonald}, there exists a matrix $ \mathbf{B}^{\prime \prime} \in R^{n \times (n-r)} $ such that $ \mathbf{B}_2 = (\mathbf{B}^\prime, \mathbf{B}^{\prime \prime}) \in R^{n \times n} $ is invertible. Therefore,
$$ \rk(\mathbf{C}\mathbf{B}) \geq \rk(\mathbf{C}\mathbf{B}^\prime) \geq \rk(\mathbf{C}\mathbf{B}_2) - (n-r) = \rk(\mathbf{C}) - (n-r). $$
In the second inequality we have used that $ \rk(\mathbf{C}\mathbf{B}_2) \leq \rk(\mathbf{C}\mathbf{B}^\prime) + \rk(\mathbf{C}\mathbf{B}^{\prime \prime}) \leq \rk(\mathbf{C}\mathbf{B}^\prime) + (n-r) $. Hence, we deduce that
$$ \dd_{SR}(\mathcal{C} \mathbf{A}) \geq \dd_{SR}(\mathcal{C}) - \rho \geq 2t+1. $$
Therefore, there exists a decoder $ D_1 : S^N \longrightarrow \mathcal{C} \mathbf{A} $ sending $ D_1(\mathbf{c}\mathbf{A} + \mathbf{e}) = \mathbf{c} \mathbf{A} $, for all $ \mathbf{c} \in \mathcal{C} $ and all $ \mathbf{e} \in S^n $ with $ \wt_{SR} (\mathbf{e}) \leq t $.

Now, let again $ \mathbf{C} \in R^{m \times n} $, $ \mathbf{B} \in R^{n \times N} $ and $ r = \frk(\mathbf{B}) $, with notation as above. Assume that $ \mathbf{C}\mathbf{B} = 0 $. Then $ \mathbf{C}\mathbf{B}^\prime = 0 $, which implies that $ \mathbf{C}\mathbf{B}_2 = \mathbf{C} (\mathbf{B}^\prime, \mathbf{B}^{\prime \prime}) = (0 , \mathbf{C}\mathbf{B}^{\prime \prime}) $. Therefore,
$$ \rk(\mathbf{C}) = \rk (\mathbf{C}\mathbf{B}_2) = \rk (\mathbf{C}\mathbf{B}^{\prime \prime}) \leq n-r. $$
Using this fact and $ \rho \leq \dd_{SR}(\mathcal{C}) $, it is easy to see that the map $ \mathcal{C} \longrightarrow \mathcal{C}\mathbf{A} $ consisting in multiplying by $ \mathbf{A} $ is injective. Hence we deduce that there exists a decoder $ D_2 : \mathcal{C} \mathbf{A} \longrightarrow \mathcal{C} $ sending $ D_2(\mathbf{c} \mathbf{A}) = \mathbf{c} $, for all $ \mathbf{c} \in \mathcal{C} $. Thus, we conclude by defining $ D_\mathbf{A} = D_2 \circ D_1 $.
\end{proof}

To conclude, we briefly describe how to decode when using a linearized Reed--Solomon code as in Definition \ref{def LRS codes}. Similarly to the proof above, we let $ \mathbf{A}^\prime \in R^{n \times r} $ be of full free rank $ r $, formed by some $ r $ columns of $ \mathbf{A} $, where $ r = \frk(\mathbf{A}) = n-\rho $. In the same way as in \cite[Sec. V-F]{secure-multishot}, if $ \mathcal{C} $ is a linearized Reed--Solomon code, then so is $ \mathcal{C} \mathbf{A}^\prime $, since $ \mathbf{A}^\prime $ has full free rank. Therefore, we may apply the decoders from Sections \ref{sec WB decoder} and \ref{sec syndrome decoder} to $ \mathcal{C} \mathbf{A}^\prime $ and recover the sent codeword $ \mathbf{c} \in \mathcal{C} $.

\ifCLASSOPTIONcaptionsoff
  \newpage
\fi



\bibliographystyle{IEEEtranS}
\end{document}